\documentclass[journal]{IEEEtran}
\usepackage{amssymb}
\usepackage{amsmath,amssymb,amsfonts}
\usepackage{multirow}
\usepackage{graphicx}
\usepackage{textcomp}
\usepackage{amsthm}

\usepackage{setspace}

\usepackage{graphicx}
\usepackage{algorithm}
\usepackage{algorithmic}
\usepackage{amsmath}
\usepackage{color}
\usepackage{amsfonts}

\newtheorem{thm}{Theorem}
\newtheorem{lem}{Lemma}

\newtheorem{defn}{Definition}

\newtheorem{exmp}{Example}
\newtheorem{rem}{Remark}

\newtheorem{obs}{Observation}

\ifCLASSINFOpdf
\else
\fi
\begin{document}
%
\title{A Multi-Feature Diffusion Model: Rumor Blocking in Social Networks}
%
%
%

\author{Jianxiong Guo,
		Tiantian Chen,
	Weili Wu,~\IEEEmembership{Member,~IEEE,}
\thanks{J. Guo T. Chen and W. Wu are with the Department
	of Computer Science, Erik Jonsson School of Engineering and Computer Science, Univerity of Texas at Dallas, Richardson, TX, 75080 USA
	
	E-mail: jianxiong.guo@utdallas.edu}
\thanks{Manuscript received April 19, 2005; revised August 26, 2015.}}

%
%

\markboth{Journal of \LaTeX\ Class Files,~Vol.~14, No.~8, August~2015}%
{Shell \MakeLowercase{\textit{et al.}}: Bare Demo of IEEEtran.cls for IEEE Journals}
%



\maketitle

\begin{abstract}
Online social networks provide a convenient platform for the spread of rumors, which could lead to serious aftermaths such as economic losses and public panic. The classical rumor blocking problem aims to launch a set of nodes as a positive cascade to compete with misinformation in order to limit the spread of rumors. However, most of the related researches were based on one-dimensional diffusion model. In reality, there are more than one feature associated with an object. The user's impression on this object is determined not just by one feature but by his/her overall evaluation on all of these features. Thus, the influence spread of this object can be decomposed into the spread of multiple features. Based on that, we propose a Multi-Feature diffusion model (MF-model) in this paper, and a novel problem, Multi-Feature Rumor Blocking (MFRB), is formulated on a multi-layer network structure according to this model. To solve MFRB, we design a creative sampling method, called Multi-Sampling, which can be applied to a multi-layer network structure. Inspired by martingale analysis, the Revised-IMM algorithm is proposed, and returns a satisfactory approximate solution to MFRB. Finally, we evaluate our proposed algorithm by conducting experiments on real datasets, and show the effectiveness and accuracy of the Revised-IMM algorithm and significantly outperforms other baseline algorithms.
\end{abstract}


%
\IEEEpeerreviewmaketitle

\section{Introduction}
%
%
%
%
\IEEEPARstart{T}{he} online social platform, such as Facebook, Twitter, LinkedIn and WeChat, have been growing rapidly over the last years, and has been a major communication platform. There are more than 1.52 billion users active daily on Facebook and 321 million users active monthly on Twitter. Usually, these social platforms can be represented as online social networks (OSNs), which is a directed graph, including individuals and their relationship. Even that providing users with convenient information exchange, OSNs provide opportunities for rumor, namely false or negative information, to spread as well. It can cause something bad happening and even panic. For example, in 2013, the fake news "President Obama is attacked" spread in Twitter caused the US stock falling wildly. Then, in 2016, the rumor made by competitors that "Hillary Clinton dumped weapons to ISIS" spread in Facebook damaged her reputation in presidential election \cite{allcott2017social}. In 2018, a video spread in Weibo that a bus fell down into river from a bridge because of a car, leading to there are 15 people losing their lives. In Weibo, All the comments were unanimously pointed to that this tragic tragedy was caused by the driving driving mistakes of the car driver. However, after investigation by the police, the disaster was brought by a dispute between the bus driver and an unreasonable passenger. Thus, the car driver was acquitted immediately.

The influence in social networks is diffused from user to user, which can be initiated by a set of seed (initial) users. The notable study of influence diffusion is traced back to Kempe et al. \cite{kempe2003maximizing} where influence maximization (IM) problem was formulated as a combinatorial optimization problem: find a subset of users as the seed set that makes the follow-up adoptions maximized by spreading the influence. They proposed two information diffusion models that are accepted by most researchers in subsequent researches: Independent Cascade model (IC-model) and Linear Threshold model (LT-model), and proved IM is NP-hard and its objective function is monotone submodular under these two models, thus, a good approximation can be obtained by natural greedy algorithm \cite{nemhauser1978analysis}. When opposite points of view, negative and positive information, from different cascades are spread at the same time on the same social network, users are more inclined to accept the information arriving on them first. Therefore, one solution of blocking rumor spread is to launch a positive cascade to compete with misinformation \cite{budak2011limiting} \cite{tong2017efficient}. Since the budget for positive seeds is limited, a classical rumor blocking (RB) problem is formulated, which aims to spread a positive cascade by selecting a positive seed set to prevent the spread of misinformation as much as possible.

The existing researches, regardless the problem about IM or RB, were based on the simple IC-model or LT-model. In other words, a piece of information that propagates through the network has only a boolean state, either good or bad. However, in the real world, the actual information diffusion is much more complicated. Let us look at an example first.
\begin{exmp}
For a computer, the features associated with this computer are price, performance, appearance and brand. Whether a user will purchase this computer is determined by his overall evaluation on these features, for example, price is high or low, performance is good or bad and so on.
\end{exmp}
\noindent
Therefore, in this paper, we propose a multi-feature diffusion model (MF-model), which matches to the realistic scenario better. For a user, the quality of a product depends on his/her evaluation on the features associated with this product. Information diffusion is not simple one-dimensional, object by object, but multi-dimensional, feature by feature. Back to the above example, someone want to promote this computer, he does not tell others directly that this computer is very good, but tell others that its price is low, performance is satisfactory and so on. In our MF-model, we assume that each feature can be diffused individually. After the diffusion of each feature terminates, users can determine whether this product is good or bad according to their own evaluation criteria. The importance of each feature is different for different users, which strengthens the generalization of model.

Then, we propose a Multi-Feature Rumor Blocking (MFRB) problem, which selects a positive seed set to compete with rumor cascade under the MF-model. The rumor from competitor is possible to spread wrong information on different features in order to lower down the reputation of the product. For example, somebody says the battery performance of iPhone is not good and its price is too expensive, or some presidential candidate's private life is extravagant. It is worth noting that although there is some negative news, this does not mean iPhone is not a good product or this presidential candidate is not qualified. The judgement for an object depends on the comprehensive evaluation on all features associated with it. Therefore, our MF-model is suitable to solve such problems. The influence spread in MF-model can be constructed in a multi-layer network structure. We prove the objective function of MFRB problem is monotone non-decreasing and submodular. Unfortunately, computing the exact influence is \#P-hard \cite{chen2010scalable}, thus the objective function is hard to compute despite Greedy algorithm is simple and effective. To estimate the expected influence, they adopt the Monte-Carlo simulation usually, but the computational cost is not acceptable. In order to improve its efficiency, the randomized algorithms based on reverse influence sampling (RIS) popularized gradually \cite{borgs2014maximizing} \cite{tang2014influence} \cite{tang2015influence}. Inspired by this idea, we propose a novel and effective sampling method, called Multi-Sampling, which can be applied to the multi-layer network structure, and we show that this sampling method is effecitve to solve our MFRB problem. Then, based on Multi-Sampling and martingale analysis, the Revised-IMM (Influence Maximization via Martingales) is formulated, whose performance for MFRB problem is as good as Greedy algorithm but much more efficient than Greedy algorithm. Besides, We can implement Revised-IMM algorithm under the different parameter settings according to your requirements for error and running time. Our contributions in this paper are summarized as follows:

\begin{enumerate}
	\item This is the first attempt to study multi-feature diffusion problem. By learning some real application scenarios, we propose MF-model to simulate multi-feature diffusion. Then, we show that MF-model can be constructed on a multi-layer network structure.
	\item MFRB problem is formulated based on MF-model, and we prove its objective function is monotone non-decreasing and submodular.
	\item We design a novel sampling method, Multi-Sampling, which can be applied to multi-layer network structure. Based on Multi-Sampling and martingale analysis, the Revised-IMM is formulated, which returns a $(1-1/e-\varepsilon)$-approximate solution of MFRB problem, and runs in $O((k+\ell)mr\log n/\varepsilon^2)$ expected time.
	\item Our proposed algorithms are evaluated on real-world datasets. The results show Revised-IMM is much faster than Greedy algorithm and almost get the same performance for MFRB problem.
\end{enumerate}

\textbf{Organiztion:} In Section \uppercase\expandafter{\romannumeral2}, we survey the related works about RB and its algorithms. We then present MF-model and MFRB problem in Section \uppercase\expandafter{\romannumeral3}, introduce our sampling technique on multi-layer network in Section \uppercase\expandafter{\romannumeral4}, and design our randomized algorithms in Section  \uppercase\expandafter{\romannumeral5}. Finally, we conduct experiments and conclude in Section  \uppercase\expandafter{\romannumeral6} and Section \uppercase\expandafter{\romannumeral7}.


\section{Related Works}
The RB problem was first proposed by Budak et al. \cite{budak2011limiting}. They presented a multi-campaign IC-model, and showed that RB can be generalized to the submodular maximization problem. Then, they proved that the objective function of RB is submodular and obtained a constant approximation ratio through greedy strategy. He et al. \cite{he2012influence} considered the competitive LT-model for RB problem and designed a $(1-1/e)$-approximation algorithm. Fan et al. \cite{fan2013least} proposed the least cost RB problem under the opportunistic one-active-one model and obtained a valid theoretical bound. Then, they considered RB problem under the time constraint, constrained by a deadline $T$ \cite{fan2014maximizing}. In addition to spreading positive cascade, there were two other methods for RB. One was protecting the most influential nodes from influenced by rumor cascade so that the influence of negative information can be reduced \cite{fan2013least} \cite{ma2016identifying} \cite{wang2013negative}. The other was removing some of relationships (edges) that play a central role in networks to limit the spread of misinformation \cite{khalil2014scalable} \cite{kimura2008minimizing} \cite{tong2012gelling}. Other researches about removing nodes or edges to block rumor, please reference \cite{ma2016identifying} \cite{wang2013negative} \cite{khalil2014scalable} \cite{kimura2008minimizing} \cite{tong2012gelling}. Please read the Srijan's comprehensive survey \cite{kumar2018false} if you are interested in more problems about misinformation.

After Kempe's seminal work \cite{kempe2003maximizing}, a large number of related researches have been done. They try to overcome the high time complexity of Greedy algorithm. It is \#P-hard \cite{chen2010scalable} \cite{chen2010scal} to compute the exact influence of a seed set under the IC-model and LT-model. Monte-Carlo simulation was adopted by many researchers to estimate the expected influence, but the computational cost was unacceptable when applied to large networks. Becasue of the low efficiency of Monte-Carlo simulation, a lot of researchers attempted improve the computational efficiency or overcome the Monte-Carlo simulations \cite{leskovec2007cost} \cite{goyal2011celf++} \cite{chen2011influence} \cite{zhang2014recent} \cite{ok2016maximizing} \cite{nguyen2017billion} \cite{li2019tiptop}. For example, Leskovec et al. proposed an CELF algorithm \cite{leskovec2007cost} with a lazy-forward evaluation, which avoids unnecessary computation by estimating the upper bound of influence. CELF++ \cite{goyal2011celf++}, an improved verson of CELF, reduced its time complexity. The effect was not satisfactory until the emergence of RIS. Reverse influence sampling (RIS) was proposed firstly by Brogs et al. \cite{borgs2014maximizing}, then a series of efficient randomized algorithm arised like TIM/TIM+ \cite{tang2014influence}, IMM \cite{tang2015influence} and SSA/D-SSA \cite{nguyen2016stop}. They were scalable methods with $(1-1/e-\varepsilon)$-approximation guarantee for the IM problem. Recently, Li et al. \cite{li2019tiptop} proposed TIPTOP based on RIS, an almost exact solutions for IM in in Billion-Scale Networks, which tried to reduces the number of samples as much as possible. Inspired by them, Tong et al. \cite{tong2017efficient} presented an efficient randomized algorithm for RB problem, whose sampling technique is called Random R-tuple. Besides, in order to improve time performance better, Tong et al. proposed a novel sampling method, hybrid sampling technique \cite{DBLP:journals/corr/abs-1901-05149}, which attached high weights to the users who are prone to be affected by rumor instead of sampling the nodes uniformly.

\section{Problem Formulation}
In this section, we introduce the MF-model and formulate the MFRB problem.
\subsection{Influence Model}
A social network can be given by a directed graph $G=(V,E)$ where $V=\{v_1,v_2,...,v_n\}$ is the set of $n$ users, $E=\{e_1,e_2,...,e_m\}$ is the set of $m$ directed edges which describe the relationship between users. The node set and edge set for graph $G$ can be referred as $V(G)$ and $E(G)$, respectively. For an edge $e=(u,v)$, $u$ is an incoming neighbor of $v$ and $v$ is an outgoing neighbor of $u$. We use $N^-(v)$ and $N^+(v)$ to denote the set of incoming neighbors and outgoing neighbors of node $v$, respectively. To simulate the diffusion process, there are two classical diffusion models, IC-model and LT-model, proposed by Kempe et al. \cite{kempe2003maximizing}. 

\begin{defn}[IC-model]
	It assumes that when a node $u$ is activated in this round, in the next round, which can execute an activation attempt to activate those inactive nodes $v$ in its outgoing neighbors $N^+(u)$ with a predefined probability. Each edge $(u,v)$ is associated with a activation probability $p_{uv}\in [0,1]$ and the activation process of different edges or different round is independent. Finally, the diffusion process stops if there is no nodes can be activated in future.
\end{defn}
\begin{defn}[LT-model]
	It assumes that each edge $(u,v)$ is associated with a weight $b_{uv}\geq0$ and each node $v$ has a threshold $\theta_v$ distributed in $[0,1]$ uniformly. For each node $v$, we require that $\Sigma_{u\in N^-(v)}b_{uv}\leq 1$, and define $A(v)$ as the set of active incoming neighbors to node $v$. The node $v$ can be activated in this round when satisfying $\Sigma_{u\in A(v)}b_{uv}\geq\theta_v$. Finally, the diffusion process stops if there is no nodes can be activated in future.
\end{defn}

Next, the monotonicity and submodularity can be defined here. We say that a set function $f:2^V\rightarrow\mathbb{R}$ is monotone if for any subsets $S\subseteq T\subseteq V$, $f(S)\leq f(T)$. A set function is submodular if for any $S \subseteq T \subseteq V$ and $u\in V \setminus T$, the marginal gain of $u$ when added to $T$ is less or equal to that when added to $S$. That is, $f(S \cup \{u\})-f(S) \geq f(T\cup \{u\}) - f(T)$.

\subsection{Realization}
Given a dircted graph $G=(V,E)$, a realization ${\rm g}=(V({\rm g}),E({\rm g}))$ is a subgraph of $G$ satisfying that $V({\rm g})=V(G)$ and $E({\rm g})\subseteq E(G)$. Under the IC-model, the diffusion probability of those edges $E({\rm g})$ in the realization ${\rm g}$ is equal to $1$. Those edges in $E({\rm g})$ are referred as to live edges, otherwise, called blocked edges. Under the IC-model, for each edge $e=(u,v)\in E(G)$, it appears in realization ${\rm g}$ with probability $p_{uv}$. Let $\Pr[{\rm g}]$ be the probability of realization ${\rm g}$ generated from $G$ under the IC-model, we have
\begin{equation}
\Pr[{\rm g}]=\prod_{e\in E({\rm g})}p_e\prod_{e\in E(G)\backslash E({\rm g})}(1-p_e)\label{e1}
\end{equation}
Obviously, there are $2^m$ possible realizations in all. The diffusion process in a realization ${\rm g}$ is a deterministic process. Thus, we can think about the propagation process from two different perspectives. Given a seed set $S$, the diffusion process can be considered as a stochastic propagation process on graph $G$, or a deterministic propagation process on a realization ${\rm g}$ generated from $G$.

In classical IM problem, we usually denote by $\sigma(S)$ the expected number of active nodes (influence) given a seed set $S$. Under the IC/LT-model, we have
\begin{equation}
\sigma(S)=\sum_{{\rm g}\in\mathcal{G}}\Pr[{\rm g}]\cdot\sigma_{{\rm g}}(S)
\end{equation}
where $\mathcal{G}$ is the set of all realizations generated from $G$ and $\sigma_{{\rm g}}(S)$ is the number of nodes for which there is a directed path of live edges from a node in $S$ in the realization ${\rm g}$. 
\begin{lem}[\cite{kempe2003maximizing}]
The objective function $\sigma(\cdot)$ is monotone non-decreasing and submodular under the IC/LT-model.
\end{lem}
\begin{rem}
The function $\sigma(\cdot)$ is a general notation to represent influence function, thus, every time we mention it, we need to emphasize which diffusion model it is based on.
\end{rem}

\subsection{Problem Definition}
First, let us consider a scenario with composed influence under a single cascade. Considering a product with $r$ features and a directed social network $G=(V,E)$, the diffusion process can be regarded as discrete steps:
\begin{enumerate}
	\item Each node represents a user, and there are two possible states associated with each user, active and inactive. The user is active when he/she is willing to purchase this product. Initially, all users are inactive.
	\item Each edge $(u,v)$ is associated with a $r$-dimensional probability vector $(p_{uv}^{1}, p_{uv}^{2},..., p_{uv}^{r})$, where $p_{uv}^{i}$ represents the activation probability of feature $i$. When user $u$ is activated, he/she will attempt to motivate his/her inactive outgoing neighbor $v$ to accept feature $i$ with probability $p_{uv}^{i}$. In this activation attempt, maybe $v$ will accept one or many features.
	\item If user $v$ receives influence from more than one active incoming neighbors simultaneously, $v$ will treat their features independently.
	\item Each user $v$ has a threshold $\theta_v$, representing the threshold that $v$ will purchase this product, and a weight vector $(w^1_v,w^2_v,...,w^r_v)$, where $w_v^i$ represents the weight of feature $i$ and $\sum_{i=1}^{r}w_v^i=1$. User $v$ will be activated if and only if the total weight of accepted features is larger than or equal to $\theta_v$.
	\item Initially, a seed set, containing initial users, is activated. At each step, every user checks whether the activated condition is satisfied. The process ends if no user becomes newly active at current step.
\end{enumerate}
\begin{obs}
	According to above composed influence model, the expected influence $\sigma(\cdot)$ (active nodes) is not submodular.
\end{obs}
\begin{proof}
	We take a counterexample to show that. Considering a product associated with five features, a user $v$ has five incoming neighbors $\{u_1,u_2,u_3,u_4,u_5\}$. For each edge $(u_i,v)$, we define $p_{u_iv}^{i}=1$ and other $p_{u_iv}^{j}=0$ for $j\neq i$. We assume that user $v$ has a threshold $\theta_v=0.5$ and weight $w_v^i=0.2$ on each feature $i$. Obviously, $\sigma(\{u_1,u_3\})-\sigma(\{u_1\})=0<\sigma(\{u_1,u_2,u_3\})-\sigma(\{u_1,u_2\})=1$ and $\{u_1\}\subseteq \{u_1,u_2\}$, contradicting the property of diminishing marginal gain. Thus, $\sigma(\cdot)$ is not submodular under the composed influence model.
\end{proof}

Are there any techniques improving the composed influence to make the expected influence be submodular? We assume user $v$ will be influenced by the features of his/her incoming neighbor $u$ only when $u$ is activated. This condition can be relaxed. Here, each feature can be spread individually, in other words, $v$ can be influenced by the accepted features of $u$, but $u$ is inactive. Thus, we can treat this relaxed diffusion model as a multi-dimensional IC-model. That is, each feature diffuses in its own dimension like the diffusion of IC-model and consults with other dimensions only when making decision to purchase the product. In order to simulate the real scene better, the threshold $\theta_v$ should be distributed in interval $[0,1]$ uniformly. In this paper, we assume that the weight for feature $i$ is equal for different users, $w^i\leftarrow w_u^i=w_v^i=...=w_z^i$. This property is useful to prove the submodularity later. So far, the revised composed influence model, called Multi-Feature Diffusion Model (MF-model), is formulated as follows:
\begin{figure}[!t]
	\centering
	\includegraphics[width=3.5in]{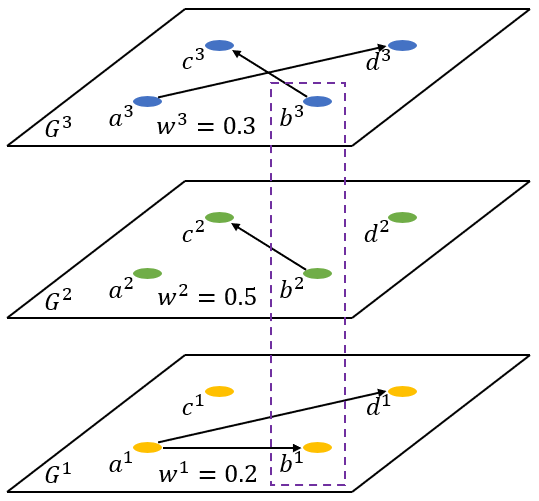}
	\caption{The form of expression to multi-layer network structure in MF-model, where each layer represents one feature and the nodes in the same column correspond to one user.}
	\label{fig_sim}
\end{figure}
\begin{defn}[MF-model]
	Given a product with $r$ features and a directed social network $G=(V,E)$, there exists an equivalent directed multi-layer graph $G'=(V',E')$. For each feature $i\in \{1,2,...,r\}$, make a copy $G^i$ of $G$. Here, we define $u^i$ in $G^i$ is the copy of corresponding node $u$ in $G$. The new graph $G'=G^1\cup G^2\cup ... \cup G^r$. For each edge $(u^i,v^i)$, the activation probability $p_{u^iv^i}$ is equal to $p_{uv}^i$ defined in composed influence model. For each layer $G^i$, only feature $i$ is spread on it and the diffusion process in this layer is independent to other layers. After all diffusion terminate, we need to determine whether a user is activated. Here, we define $x(v^i)=1$ if node $v^i$ accepts feature $i$, otherwise $x(v^i)=0$. If user $v$ satisfies the following condition:
	\begin{equation}
	\sum_{i=1}^{r}w^i\cdot x(v^i)\geq \theta_v
	\end{equation}
	we say this user $v$ is activated. Other definition is similar to that in composed influence model.
\end{defn}
\begin{rem}
The nodes in $V(G)$ are called user node, or user; but the nodes in $V(G')$ are called feature node, or feature. For example, a user node $u$ corresponds to feature node set $\{u^1,u^2,...,u^r\}$. A user node can be activated when satisfying Equation (3). To avoid confusion, we say a feature node is accepted when it is activated in its layer.
\end{rem}

Then, we take an example, shown as Fig. \ref{fig_sim}, to demonstrate how MF-model works. The example in Fig. \ref{fig_sim} is a realization of multi-layer graph $G'$, and $G^1,G^2,G^3$ corresponds to three features. Initially, user $b$ is activated, namely features $\{b^1,b^2,b^3\}$ are accepted. After diffustion stops, for user $c$, feature $c^2$ and $c^3$ are accepted. Assuming $\theta_c=0.5$, we have $\sum_{i=1}^{3}w^i\cdot x(c^i)=0.2\times 0+0.5\times 1+0.3\times 1=0.8>\theta_c$, thus, user $c$ is activated.

From now, we consider the MF-model as information diffusion model, and there are multiple cascades diffusing on the same social network. A user is $C$-active if he/she is activated by cascade $C$. Initially, all users are $\emptyset$-active. Shown as Definition 3, each feature diffuses independently, and then, we are able to determine whether the user is active after all feature diffusions have terminated. Let us consider the following scenario: there are two cascades spreading on network $G=(V,E)$, the positive cascade $C_p$ and the negative (rumor) cascade $C_r$. Given the rumor seed set $S_r$, we want to launch a positive cascade to compete against the rumor cascade. Denote by $S_p$ the seed set of positive cascade. The information from $S_r$ and $S_p$ diffuses simultaneously under the MF-model. On the layer $G^i$, if two opposite cascades activate a node $v^i$ successfully at the same time, rumor cascade has a higher priority. That is, $v^i$ will be $C_r$-accepted. After all feature diffusions have terminated, we are able to determine whether a user is $C_r$-active or $C_p$-active.

\begin{rem}
	For a user $u$, we define $\mathcal{F}(u)=\{u^1,u^2,...,u^r\}$ as $u$'s corresponding feature nodes. Assuming that a seed set $S$ is served for cascade $C$, we say $S$ is partially $C$-active if there exists some user $u\in S$, only part of feature nodes in $\mathcal{F}(u)$ accept cascade $C$. For example, only $\{u^1,u^3\}\subset\mathcal{F}(u)$ accept cascade $C$. On the contrary, $S$ is fully $C$-active if all feature nodes of each user in $S$ accept cascade $C$. Then, we denote by $S^i$, $i\in\{1,2,...,r\}$, the set of corresponding feature nodes of $S$ in layer $i$ that accept cascade $C$. If $S$ is partially $C$-active, then $|S^i|\leq|S|$. If $S$ is fully $C$-active, then $|S^i|=|S|$.
\end{rem}

In the real world, a user can hardly be so stupid that he/she believes the rumor that announces all the features of a product are not good. Thus, we assume that rumor seed set $S_r$ is partially $C_r$-active, in other words, there exists some user $u\in S_r$ who does not believe all the features of this product are bad when rumor is this product is totally bad. And positive seed set $S_p$ is fully $C_p$-active. A user is $\bar{C_r}$-active if he/she is not $C_r$-active. For user $v$, we define $r(v^i)=1$ if node $v^i$ accepts rumor cascade, otherwise $r(v^i)=0$. If user $v$ satisfies the following condition:
\begin{equation}
\sum_{i=1}^{r}w^i\cdot (1-r(v^i))\geq \theta_v
\end{equation}
this user $v$ is not activated by rumor cascade. If Inequality (4) is satisfied, we say this user $v$ is $\bar{C_r}$-active. Besides, we denote by $f(S_p)$ the expected number of $\bar{C_r}$-active users given a positive seed set $S_p$. So far, the Multi-Feature Rumor Blocking (MFRB) problem is formulated.
\begin{defn}[MFRB]
	Given a social network $G=(V,E)$, a positive integer $k$ and a partially $C_r$-active rumor set $S_r$, MFRB selects an fully $C_p$-active positive seed set $S_p^\circ$, $|S_p^\circ|\leq k$, from $V(G)\backslash S_r$ to make the expected number of $\bar{C}_r$-active users $f(S_p)$ maximized under the MF-model. We have
	\begin{equation}
	S_p^\circ=\arg \max_{S_p\subseteq V\backslash S_r, |S_p|\leq k}f(S_p)
	\end{equation}
\end{defn}
\begin{thm}
	In MFRB problem, the expected number of $\bar{C_r}$-active users $f(S_p)$ is monotone non-decreasing and submodular with respect to $S_p$.
\end{thm}
\begin{proof}
	In order to prove this theorem, we need to represent $f(S_p)$ mathmatically, firstly. The $f(S_p)$ under the MF-model can be defined as follows:
	\begin{equation}
	f(S_p)=\sum_{i=1}^{r}w^i\sum_{{\rm g}^i\in \mathcal{G}^i}\Pr[{\rm g}^i]\cdot f_{{\rm g}^i}(S_p^i)
	\end{equation}
	where $f_{{\rm g}^i}(S_p^i)$ is the number of feature nodes that cannot be reached by rumor cascade from $S_r^i$ in the realization ${\rm g}^i$ of graph $G^i$, and $S_p^i$ is the set of feature nodes in layer $i$ corresponding to users in $S_p$ according to fully active assumption of $S_p$.
	
	Then, $\sum_{{\rm g}^i\in \mathcal{G}^i}\Pr[{\rm g}^i]\cdot f_{{\rm g}^i}(S_p^i)$ is the average number of feature nodes, which is $\bar{C_r}$-accepted in feature $i$. Becuase the threshold $\theta_v$ is uniformly distributed in $[0,1]$ and $\sum_{i=1}^{r}w^i=1$, each $\bar{C}_r$-active node $u^i$ contributes $w^i$ to the expectation of $\bar{C}_r$-accepted users. In other words, the probability of user $u$ terminated as $\bar{C}_r$-active increases by $w^i$, so $f(S_p)$ increases by $w^i$. In addition, $f_{{\rm g}^i}(S^i)$ is monotone non-decreasing and submodular, which has been proven by Tong et al. \cite{tong2017efficient}. $f(S_p)$ is a linear combination of $f_{{\rm g}^i}(S_p^i)$, thus, $f(S_p)$ is monotone and submodular with respect to $S_p$.
\end{proof}

\section{Sampling Technique}
In last section, we have proven that the MFRB problem is monotone non-decreasing and submodular, thus, the simple greedy algorithm can get a $(1-1/e)$-approximation solution \cite{nemhauser1978analysis}. However, the computational cost of greedy algorithm is too high because computing the objective function of MFRB is \#P-hard \cite{chen2010scalable}. Therefore, it is not advisable to compute $f(S_p)$ directly, instead of that, we can find an estimator of $f(S_p)$ by some sampling techniques, and then make this estimator maximized. Here, we will get help from Random R-tuple sampling technique, provided by Tong et al. \cite{tong2017efficient}, to design our estimator. First, we define the expected $\bar{C_r}$-accepted feature nodes in layer $i$ as $f^i(S^i_p)$:
\begin{equation}
f^i(S^i_p)=\sum_{{\rm g}^i\in \mathcal{G}^i}\Pr[{\rm g}^i]\cdot f_{{\rm g}^i}(S_p^i)
\end{equation}

\begin{algorithm}[!t]
	\caption{\textbf{R-sampling $({\rm g}^i,v^i,S_r^i)$}}\label{a1}
	\begin{algorithmic}[1]
		\renewcommand{\algorithmicrequire}{\textbf{Input:}}
		\renewcommand{\algorithmicensure}{\textbf{Output:}}
		\REQUIRE ${\rm g}^i=(V^i,E^i({\rm g}))$, $v^i$ and $S_r^i$
		\ENSURE $V^*$ or $V^i$
		\STATE Initialize $V_{cur}\leftarrow\{v^i\}$
		\STATE Initialize $V^*\leftarrow\emptyset$
		\STATE Initialize an empty queue $Q$
		\WHILE {true}
		\IF {$V_{cur}=\emptyset$}
		\STATE Return $V^i$
		\ENDIF
		\IF {$V_{cur}\cap S_r^i\neq\emptyset$}
		\STATE Return $V^*$
		\ENDIF
		\STATE $V^*\leftarrow V^*\cup V_{cur}$
		\STATE $Q\leftarrow Q\cup V_{cur}$
		\STATE $V_{cur}\leftarrow \emptyset$
		\WHILE {$Q\neq \emptyset$}
		\STATE $u^i=Q$.pop()
		\STATE $V_{cur}\leftarrow V_{cur}\cup \{t^i|t^i\in N^-(u^i)$ and $t^i\notin V^*\}$
		\ENDWHILE
		\ENDWHILE
	\end{algorithmic}
\end{algorithm}

For any feature node $v^i\in V(G^i)$, we use R-tuple sampling technique \cite{tong2017efficient} on graph $G^i=(V^i,E^i)$ given rumor accepted set $S_r^i$, here, we call it as R-sampling. Given ${\rm g}^i=(V^i,E^i({\rm g}))$ as a realization of $G^i$, feature node $v^i$ and rumor accepted set $S^i_r$, the R-sampling is shown as Algorithm \ref{a1}, which is a little different from the original version in \cite{tong2017efficient}. The R-sampling starts from $v^i$ in $V^*$ and determine whether the incoming neighbors of the nodes in $V^*$ can be added to $V^*$ in a breadth-first searching until one of the rumor nodes in $S_r^i$ is reached or no node can be furthered reached. Then, the random R-sampling in graph $G^i$ can be generated by the following steps:
\begin{enumerate}
	\item Select a node $v^i$ from $V(G^i)$ uniformly.
	\item Generate a realization ${\rm g}^i$ of $G^i$.
	\item Get an R-sampling $V^*$ returned by Algorithm 1, R-sampling $({\rm g}^i,v^i,S_r^i)$
\end{enumerate}
This process, called Single-Sampling, is shown as Algorithm \ref{a2}. Intuitively, $R^i$ contains the feature nodes that could prevent $v^i$ in ${\rm g}^i$ from influenced by rumor set $S_r^i$ when one of them accepts positive cascade. For any positive seed set $S_p$, we define:
\begin{equation}
x(S_p^i,R^i)=\left
\{\begin{IEEEeqnarraybox}[\relax][c]{l's}
1,&if $S_p^i\cap R^i\neq\emptyset$\\
0,&otherwise%
\end{IEEEeqnarraybox}
\right.
\end{equation}
\begin{rem}
	For convenience, we can consider positive seed set as $S_p=S_p^1\cup S_p^2\cup ... \cup S_p^r$ and rumor seed set as $S_r=S_r^1\cup S_r^2\cup ... \cup S_r^r$.
\end{rem}

\begin{algorithm}[!t]
	\caption{\textbf{Single-Sampling $(G^i,S_r^i)$}}\label{a2}
	\begin{algorithmic}[1]
		\renewcommand{\algorithmicrequire}{\textbf{Input:}}
		\renewcommand{\algorithmicensure}{\textbf{Output:}}
		\REQUIRE $G^i=(V^i,E^i)$ and $S_r^i$
		\ENSURE $R^i$
		\STATE Select a node $v^i$ from $V^i$ uniformly.
		\STATE Generate a realization ${\rm g}^i$ of $G^i$.
		\STATE $R^i\leftarrow$ R-sampling $({\rm g}^i,v^i,S_r^i)$
		\STATE Return $R^i$
	\end{algorithmic}
\end{algorithm}

\noindent
Here, it is easy to know that $x(S_p,R^i)=x(S_p^i,R^i)$ because $S_p^j\cap R^i=\emptyset$ when $i\neq j$. Under the set $S_r^i$, we generate a collection of Single-Sampling $\mathcal{R}^i=\{R^i_1,R^i_2,...,R^i_\pi\}$ given the feature $i$. We define $F_{\mathcal{R}^i}(S_p^i)$, the fraction of Single-Sampling in $\mathcal{R}^i$ covered by $S_p^i$, as follows:
\begin{equation}
F_{\mathcal{R}^i}(S_p^i)=\frac{1}{\pi}\cdot \sum_{j=1}^{\pi}x(S_p^i,R^i_j)
\end{equation}
\begin{lem}[\cite{tong2017efficient}]
	Given $G^i=(V^i,E^i)$ and $S_r^i$ for feature $i$, we have $\mathbb{E}[n\cdot F_{\mathcal{R}^i}(S_p^i)]=f^i(S_p^i)$ for $S_p^i\subseteq V^i\backslash S_r^i$.
\end{lem}
So far, we have obtained an unbiased estimator for $f^i(S_p^i)$, but it cannot be applied to solve our FMRB problem directly because multiple features exist in our problem. We can consider this problem in another way. Given $G'=(V',E')$ and rumor seed set $S_r$, $V'=V^1\cup V^2\cup...\cup V^r$, we select a feature node $v\in V'$ from these $nr$ nodes randomly. After confirming this feature node we select belongs to feature $i$, we generate a realization ${\rm g}^i$ of $G^i$ and then get a R-sampling $R$ returned by Algorithm \ref{a1}. We call this process as Multi-Sampling, which is shown as Algorithm \ref{a3}. Let $\mathcal{R}$ be a collection of Multi-Samplings, $\mathcal{R}=\{R_1,R_2,...,R_\theta\}$, that contains $\theta$ Multi-Samplings. We define $W_{\mathcal{R}}(S_p)$, the weighted average fraction of Multi-Samplings in $\mathcal{R}$ covered by $S_p$, as follows:
\begin{equation}
W_{\mathcal{R}}(S_p)=\frac{1}{\theta}\cdot\sum_{i=1}^{r}w^i \sum_{j=1}^{\theta}x(S_p^i,R_j)
\end{equation}
\begin{thm}
	Given $G=(V,E)$ and rumor seed set $S_r$, we have $\mathbb{E}[nr\cdot W_{\mathcal{R}}(S_p)]=f(S_p)$ for $S_p\subseteq V\backslash S_r$.
\end{thm}
\begin{proof}
	In Algorithm \ref{a3}, we select a node $v$ from $V'$ uniformly, which means that the average number of Multi-Samplings in $\mathcal{R}$ generated by a node in each feature $i$ is the same. We define the number of Multi-Samplings in $\mathcal{R}$ generated by a node in feature $i$ as $N_\mathcal{R}(i)$, thus, $\mathbb{E}[N_\mathcal{R}(i)]=\theta/r$ for $i\in\{1,2,...,r\}$. Therefore, $\mathbb{E}[F_{\mathcal{R}^i}(S_p^i)]$ can be expressed as
	\begin{equation}
	\mathbb{E}[F_{\mathcal{R}^i}(S_p^i)]=r\cdot \mathbb{E}[F_{\mathcal{R}}(S_p^i)]
	\end{equation}
	According to Equation (10) and (11), for $\mathbb{E}[nr\cdot W_{\mathcal{R}}(S_p)]$, we have the following observation:
	\begin{flalign}
		\mathbb{E}[nr\cdot W_{\mathcal{R}}(S_p)]&=\sum_{i=1}^{r}w^i\cdot \bigg(nr\cdot \mathbb{E}\Big[\frac{1}{\theta}\sum_{j=1}^{\theta}x(S_p^i,R_j)\Big]\bigg)\nonumber\\
		&=\sum_{i=1}^{r}w^i\cdot \big(nr\cdot \mathbb{E}[F_{\mathcal{R}}(S_p^i)]\big)\nonumber\\
		&=\sum_{i=1}^{r}w^i\cdot \big(n\cdot \mathbb{E}[F_{\mathcal{R}^i}(S_p^i)]\big)\nonumber\\
		&=\sum_{i=1}^{r}w^i\cdot f^i(S^i_p)\nonumber\\
		&=f(S_p)\nonumber
	\end{flalign}
	From above, we know that $nr\cdot W_{\mathcal{R}}(S_p)$ is an unbiased estimator to $f(S_p)$. Then, the theorem is proved.
\end{proof}

\begin{algorithm}[!t]
	\caption{\textbf{Multi-Sampling $(G,S_r)$}}\label{a3}
	\begin{algorithmic}[1]
		\renewcommand{\algorithmicrequire}{\textbf{Input:}}
		\renewcommand{\algorithmicensure}{\textbf{Output:}}
		\REQUIRE $G=(V,E)$ and $S_r$
		\ENSURE $R$
		\STATE Select a node $v$ from $V^1\cup V^2\cup...\cup V^r$ uniformly.
		\STATE Confirm $v\in V^i$
		\STATE Generate a realization ${\rm g}^i$ of $G^i$.
		\STATE $R\leftarrow$ R-sampling $({\rm g}^i,v^i,S_r^i)$
		\STATE Return $R$
	\end{algorithmic}
\end{algorithm}

\section{The Algorithm}
From the last section, $nr\cdot W_{\mathcal{R}}(S_p)$ over $\mathcal{R}$ can be used as an unbiased estimator of objective function $f(S_p)$. Before designing our algorithm, we need to introduce martingale and its relative properties first, defined as follows:
\begin{defn}[Martingale \cite{chung2006concentration}]
	A martingale is a sequence of random variables $Y_1,Y_2,Y_3,...$, such that $\mathbb{E}[|Y_i|]<+\infty$ and $\mathbb{E}[Y_i|Y_1,Y_2,...,Y_{i-1}]=Y_{i-1}$ for any $i$.
\end{defn}
Consider a collection of Multi-Samplings, $\mathcal{R}=\{R_1,R_2,...,R_\theta\}$. Let $p=f(S_p)/nr$, we define $M_k$ as
\begin{equation}
M_k=\sum_{j=1}^{k}\bigg(\sum_{i=1}^{r}w^i\cdot x(S_p^i,R_j)-p\bigg)
\end{equation}
where $k=\{1,2,...,\theta\}$. Becasue of the linearity of expectation, $p=\mathbb{E}[W_{\mathcal{R}}(S_p)]=\mathbb{E}[\sum_{i=1}^{r}w^i\cdot x(S_p^i,R_j)]$, we have $\mathbb{E}[M_i]=0$ and $\mathbb{E}[|M_i|]<+\infty$. The value of $x(S_p,R_j)$ is independent to the value from $x(S_p,R_1)$ to $x(S_p,R_{j-1})$, thus, $\mathbb{E}[M_i|M_1,M_2,...,M_{i-1}]=M_{i-1}$. Therefore, $M_1,M_2,...,M_\theta$ is a martingale.
\begin{lem}[\cite{chung2006concentration}]
	Let $Y_1,Y_2,Y_3,...$ be a martingale, such that $|Y_1|\leq a$, $|Y_j-Y_{j-1}|\leq a$ for each $j\in \{2,...,i\}$ and $Var[Y_1]+\sum_{j=2}^{\theta}Var[Y_j|Y_1,Y_2,...,Y_{j-1}]<=b$. Then for any $\gamma>0$, we have
	\begin{equation}
	\Pr[Y_i-\mathbb{E}[Y_i]\leq-\gamma]\leq \exp\bigg(-\frac{\gamma^2}{2b}\bigg)
	\end{equation}
	\begin{equation}
	\Pr[Y_i-\mathbb{E}[Y_i]\geq\gamma]\leq \exp\bigg(-\frac{\gamma^2}{\frac{2}{3}a\gamma+2b}\bigg)
	\end{equation}
\end{lem}
Considering the martingale $M_1,M_2,...,M_\theta$, we can set $a=1$ because $|M_1|\leq 1$ and $|M_j-M_{j-1}|\leq 1$ for each $j\in\{2,...,\theta\}$. Here, we define the maximum weight $\bar{w}$ over all features as $\bar{w}=\max\{w_1,w_2,...,w_r\}$. Obviously, we have $\sum_{i=1}^{r}w^i\cdot x(S_p^i,R_j)\leq \bar{w}\cdot x(S_p,R_j)$ because for each Multi-Sampling $R_j$, which can only be covered by one kind of feature nodes. If $x(S_p^y,R_j)=1$, then we have $x(S_p^z,R_j)=0$ for $z\in\{1,2,...,r\}\backslash\{y\}$. Based on the properties of variance and Equation (12), we can set $b=\bar{w}\cdot p\theta$ because
\begin{flalign}
	&Var[M_1]+\sum_{j=2}^{\theta}Var[M_j|M_1,M_2,...,M_{j-1}]\nonumber\\
	&=\sum_{j=1}^{\theta}Var[\sum_{i=1}^{r}w^i\cdot x(S_p^i,R_j)]\nonumber\\
	&=\sum_{j=1}^{\theta}\{\mathbb{E}[(\sum_{i=1}^{r}w^i\cdot x(S_p^i,R_j))^2]-(\mathbb{E}[\sum_{i=1}^{r}w^i\cdot x(S_p^i,R_j)])^2\big\}\nonumber\\
	&=\sum_{j=1}^{\theta}\{\mathbb{E}[\sum_{i=1}^{r}(w^i)^2\cdot x(S_p^i,R_j)]-p^2\big\}\\
	&\leq \sum_{j=1}^{\theta}\{\mathbb{E}[\sum_{i=1}^{r}w^i\cdot x(S_p^i,R_j)]\big\}\cdot \bar{w}\nonumber\\
	&=\bar{w}\cdot p\theta \nonumber
\end{flalign}
where the Inequality (15) holds because of the above analysis. If $x(S_p^y,R_j)=1$, then we have $x(S_p^z,R_j)=0$ for $z\in\{1,2,...,r\}\backslash\{y\}$. Thus, $\sum_{i=1}^{r}w^i\cdot x(S_p^i,R_j)=w^y\cdot x(S_p^y,R_j)$, so $(w^y\cdot x(S_p^y,R_j))^2=\sum_{i=1}^{r}(w^i)^2\cdot x(S_p^i,R_j)$. Then, we have following two inequality for any $\varepsilon>0$ according to Equation (13) (14):
\begin{equation}
\begin{aligned}
&\Pr\bigg[\sum_{j=1}^{\theta}\sum_{i=1}^{r}w^i\cdot x(S_p^i,R_j)-p\theta\leq -\varepsilon\cdot p\theta\bigg]\\
&\leq \exp\bigg(-\frac{\varepsilon^2}{2\bar{w}}\cdot p\theta\bigg)
\end{aligned}
\end{equation}
\begin{equation}
\begin{aligned}
&\Pr\bigg[\sum_{j=1}^{\theta}\sum_{i=1}^{r}w^i\cdot x(S_p^i,R_j)-p\theta\geq \varepsilon\cdot p\theta\bigg]\\
&\leq\exp\bigg(-\frac{\varepsilon^2}{2\bar{w}+\frac{2}{3}\varepsilon}\cdot p\theta\bigg)
\end{aligned}
\end{equation}

Borrowed from the idea of IMM algorithm \cite{tang2015influence}, our solution of MFRB problem can be designed, which consists of two stages as follows:

\begin{enumerate}
	\item Sampling Multi-Sampling: This stage generates Multi-Sampling iteratively and put them into $\mathcal{R}$ until satisfying a certain stopping condition.
	\item Node selection: This stage adopts greedy algorithm to drive a size-k user set $S_p$ that covers sub-maximum weight of Multi-Samplings in $\mathcal{R}$.
\end{enumerate}

\subsection{Node Selection}
Let $\mathcal{R}=\{R_1,R_2,...,R_\theta\}$ be a collection of Multi-Samplings and $W_{\mathcal{R}}(S_p)$ be the weighted average fraction of Multi-Samplings in $\mathcal{R}$ covered by $S_p$. The node selection stage is shown in Algorithm \ref{a4}. Here, we define the optimal solution as $S_p^\circ$ and optimal value as ${\rm OPT}=f(S_p^\circ)$. Because $W_{\mathcal{R}}(\cdot)$ is monotone non-decreasing and submodular, which guarantees that $W_{\mathcal{R}}(S_p^*)$ returned by Algorithm \ref{a4} satisfies $W_{\mathcal{R}}(S_p^*)\geq(1-1/e)\cdot W_{\mathcal{R}}(S_p^\circ)$.
\begin{algorithm}[!t]
	\caption{\textbf{NodeSelection $(\mathcal{R},k)$}}\label{a4}
	\begin{algorithmic}[1]
		\renewcommand{\algorithmicrequire}{\textbf{Input:}}
		\renewcommand{\algorithmicensure}{\textbf{Output:}}
		\REQUIRE $\mathcal{R}=\{R_1,R_2,...,R_\theta\}$ and $k$
		\ENSURE $\{S_p^*$, $W_{\mathcal{R}}(S_p^*)\}$
		\STATE Initialize $S_p^*\leftarrow\emptyset$
		\FOR {$1$ to $k$}
		\STATE $u=\arg\max_{u\in V\backslash S_r}(W_{\mathcal{R}}(S_p^*\cup \{u\})-W_{\mathcal{R}}(S_p^*))$
		\STATE $S_p^*=S_p^*\cup \{u\}$
		\ENDFOR
		\STATE Return $\{S_p^*$, $W_{\mathcal{R}}(S_p^*)\}$
	\end{algorithmic}
\end{algorithm}
\begin{lem}
	Given rumor seed set $S_r$, $W_{\mathcal{R}}(S_p)$ is monotone non-decreasing and submodular with respect to $S_p$.
\end{lem}
\begin{proof}
	First, we show $W_{\mathcal{R}}(\cdot)$ is monotone non-decreasing. For any positive seed set $S_p\subseteq V\backslash S_r$ and node $u\not\subseteq S_p\cup S_r$, we have
	\begin{equation}
	\begin{aligned}
	&W_{\mathcal{R}}(S_p\cup \{u\})-W_{\mathcal{R}}(S_p)\\
	&=\frac{1}{\theta}\cdot\sum_{i=1}^{r}w^i \sum_{j=1}^{\theta}(x(S_p^i\cup\{u^i\},R_j)-x(S_p^i,R_j))
	\end{aligned}
	\end{equation}
	It is monotone non-decreasing becuase $x(S_p^i,R_j)=1$ implies $x(S_p^i\cup\{u^i\},R_j)=1$, $W_{\mathcal{R}}(S_p\cup \{u\})-W_{\mathcal{R}}(S_p)\geq 0$. Next, we show $W_{\mathcal{R}}(\cdot)$ is submodular. Given any $S_{p1}\subseteq S_{p2}\subseteq V\backslash S_r$ and $u\not\subseteq S_{p2}\cup S_r$, it is equivalent to prove $x(S_{p1}^i\cup\{u^i\},R_j)-x(S_{p1}^i,R_j)\geq x(S_{p2}^i\cup\{u^i\},R_j)-x(S_{p2}^i,R_j)$ according to Equation (18). Here, we need to show that $x(S_{p1}^i\cup\{u^i\},R_j)-x(S_{p1}^i,R_j)=1$ whenever $x(S_{p2}^i\cup\{u^i\},R_j)-x(S_{p2}^i,R_j)=1$, which implies $x(S_{p2}^i\cup\{u^i\},R_j)=1$ and $x(S_{p2}^i,R_j)=0$. $x(S_{p2}^i,R_j)=0$ means that $S_{p2}^i\cup R_j=\emptyset$ and $S_{p1}^i\cup R_j=\emptyset$ because of $S_{p1}\subseteq S_{p2}$. Then, $x(S_{p2}^i\cup\{u^i\},R_j)=1$ means that $\{u^i\}\cup R_j\neq\emptyset$,  so $x(S_{p1}^i\cup\{u^i\},R_j)=1$. Therefore, $x(S_{p1}^i\cup\{u^i\},R_j)-x(S_{p1}^i,R_j)=1$ and $W_{\mathcal{R}}(\cdot)$ is submodular, the Lemma is proved.
\end{proof}
\begin{lem}
	If the number of Multi-Samplings $\theta$ in $\mathcal{R}$ of Algorithm \ref{a4} satisfies that $\theta\geq\theta_1$,
	\begin{equation}
	\theta_1=\frac{2nr\bar{w}\cdot \log(1/\delta_1)}{\varepsilon_1^2\cdot {\rm OPT}}
	\end{equation}
	then, $nr\cdot W_{\mathcal{R}}(S_p^*)\geq (1-1/e)(1-\varepsilon_1)\cdot {\rm OPT}$ holds
	with at least $1-\delta_1$ probability.
\end{lem}
\begin{proof}
	See Appendix A.
\end{proof}
\begin{lem}
	If the number of Multi-Samplings $\theta$ in $\mathcal{R}$ of Algorithm \ref{a4} satisfies that $\theta\geq\theta_2$,
	\begin{equation}
	\theta_2=\frac{(2\bar{w}+\frac{2}{3}\varepsilon_2)nr\cdot\log\big(\binom{n-n_r}{k}/\delta_2\big)}{\varepsilon_2^2\cdot {\rm OPT}}
	\end{equation}
	then, $nr\cdot W_{\mathcal{R}}(S_p^*)-f(S_p^*)\leq\varepsilon_2\cdot {\rm OPT}$ holds
	with at least $1-\delta_2$ probability, where $n_r=|S_r|$.
\end{lem}
\begin{proof}
See Appendix B.
\end{proof}
\begin{thm}
	Given any $\varepsilon_1<\varepsilon$, $\varepsilon_2=\varepsilon-(1-1/e)\varepsilon_1$ and $\delta_1,\delta_2\in(0,1)$ with $\delta_1+\delta_2\leq1/n^\ell$, if the number of Multi-Samplings $\theta$ in $\mathcal{R}$ of Algorithm \ref{a4} satisfies that $\theta\geq\max\{\theta_1,\theta_2\}$, it returns a $(1-1/e-\varepsilon)$-approximate solution of MFRB problem with at least $1-1/n^\ell$ probability.
\end{thm}
\begin{proof}
	By Lemma 4 and Lemma 5, they hold with $(1-\delta_1)(1-\delta_2)>1-(\delta_1+\delta_2)\geq 1-1/n^\ell$ probability. Then, $f(S_p^*)\geq nr\cdot W_{\mathcal{R}}(S_p^*)-\varepsilon_2\cdot {\rm OPT}\geq (1-1/e)(1-\varepsilon_1)\cdot {\rm OPT}-\varepsilon_2\cdot {\rm OPT}=(1-1/e-((1-1/e)\varepsilon_1+\varepsilon_2))\cdot {\rm OPT}=(1-1/e-\varepsilon)\cdot {\rm OPT}$. The Theorem is proved.
\end{proof}

From Theorem 3, we need to compute $\theta\geq\max\{\theta_1,\theta_2\}$ and ensure $\mathcal{R}$ contains at least $\theta$ Multi-Samplings. In order to derive such a $\theta$, which is feasible to find the minimum $\theta$. Here, we set $\delta_1=\delta_2=1/(2n^\ell)$ and $\varepsilon_1=\varepsilon_2=\varepsilon/(2-1/e)$ such that $\varepsilon_2=\varepsilon-(1-1/e)\varepsilon_1$. We define $\lambda^*$ as
\begin{equation}
\lambda^*=\frac{2nr\bar{w}\left(2-\frac{1}{e}\right)\left(2-\frac{1}{e}+\frac{\varepsilon}{3\bar{w}}\right)\left(\log\left(\tbinom{n-n_r}{k}\cdot 2n^\ell\right)\right)}{\varepsilon^2}
\end{equation}
and $\theta^*=\lambda^*/{\rm OPT}$. We can verify $\theta^*\geq\max\{\theta_1,\theta_2\}$ easily. However, it is difficult to compute the value of ${\rm OPT}$ directly. In the next subsection, we will find a lower bound ${\rm LB}$ of optimal value instead of ${\rm OPT}$ and determine the number of Multi-Samplings in $\mathcal{R}$ by $\lambda^*/{\rm LB}$.

\begin{algorithm}[!t]
	\caption{\textbf{Sampling $(G,k,r,\varepsilon,\ell)$}}\label{a5}
	\begin{algorithmic}[1]
		\renewcommand{\algorithmicrequire}{\textbf{Input:}}
		\renewcommand{\algorithmicensure}{\textbf{Output:}}
		\REQUIRE $G=(V,E)$, parameters $k$, $r$, $\varepsilon$ and $\ell$
		\ENSURE A collection $\mathcal{R}$
		\STATE Initialize $\mathcal{R}=\emptyset$, ${\rm LB}=1$, $\varepsilon'=\sqrt{2}\varepsilon$
		\STATE Initialize $\mathcal{R}'=\emptyset$
		\STATE $\lambda'=nr\left(2\bar{w}+\frac{2}{3}\varepsilon'\right)\left(\log\binom{n-n_r}{k}/\delta_3\right)\varepsilon'^{-2}$
		\STATE $\lambda^*=2nr\bar{w}\left(2-\frac{1}{e}\right)\left(2-\frac{1}{e}+\frac{\varepsilon}{3\bar{w}}\right)\left(\log\left(\tbinom{n-n_r}{k}\cdot 2n^\ell\right)\right)\varepsilon^{-2}$
		\FOR {$i=1$ to $\log_2(nr)-1$}
		\STATE $x_i=nr\cdot2^{-i}$
		\STATE $\theta_i=\lambda'/x_i$
		\WHILE {$|\mathcal{R}|\leq\theta_i$}
		\STATE $R\leftarrow$ Multi-Sampling $(G,S_r)$
		\STATE $\mathcal{R}=\mathcal{R}\cup R$
		\ENDWHILE
		\STATE $\{S_i, W_{\mathcal{R}}(S_i)\}\leftarrow$ NodeSelection $(\mathcal{R},k)$
		\IF {$nr\cdot W_{\mathcal{R}}(S_i)\geq(1+\varepsilon')\cdot x_i$}
		\STATE ${\rm LB}=nr\cdot W_{\mathcal{R}}(S_i)/(1+\varepsilon')$
		\STATE break
		\ENDIF
		\ENDFOR
		\STATE $\theta\leftarrow\lambda^*/{\rm LB}$
		\WHILE {$|\mathcal{R}'|\leq\theta$}
		\STATE $R\leftarrow$ Multi-Sampling $(G,S_r)$
		\STATE $\mathcal{R}'=\mathcal{R}'\cup R$
		\ENDWHILE
		\STATE Return $\mathcal{R}'$
	\end{algorithmic}
\end{algorithm}

\subsection{Sampling Multi-Sampling}
In last subsection, we have obtained the approximate minimum value of $\theta$. Inspired by the idea of IMM algorithm \cite{tang2015influence}, we aim to make the difference between LB and {\rm OPT} as close as possible. The process of Sampling Multi-Sampling stage is shown in Algorithm 5. In iteration $i$, we generate a certain number of Multi-Samplings, put them into $\mathcal{R}$ and call Algorithm \ref{a4}, then compare this result $W_{\mathcal{R}}(S_i)$ with statistical test $(1+\varepsilon')\cdot x_i$. When the LB is close to {\rm OPT} enough, it terminates the for-loop with a high probability. Obviously, the Multi-Samplings generated by Algorithm \ref{a5} are not independent, because those Multi-Samplings generated in $i^{th}$ iteration are determined by whether the size of collection $\mathcal{R}$ in $(i-1)^{th}$ iteration is large enough to make the estimation accurate. It can be analyzed by use of martingale technique, which is shown as Lemma 7 and Lemma 8. Finally, we generate a new collection of Multi-Samplings, and we will explain why we need to do that later.

\begin{lem}
	Consider the $i^{th}$ iteration in Algorithm 5, if the number of Multi-Samplings $\theta_i$ in $\mathcal{R}$ satisfies
	\begin{equation}
	\theta_i\geq\frac{\left(2\bar{w}+\frac{2}{3}\varepsilon'\right)nr\cdot\left(\log\binom{n-n_r}{k}/\delta_3\right)}{\varepsilon'^{2}\cdot x_i}
	\end{equation}
	If ${\rm OPT}<x_i$, then $nr\cdot W_\mathcal{R}(S_i)<(1+\varepsilon')\cdot x_i$ holds with at least $1-\delta_3$ probability.
\end{lem}
\begin{proof}
See Appendix C.
\end{proof}
\begin{lem}
	Consider the $i^{th}$ iteration in Algorithm \ref{a5}, if ${\rm OPT}\geq x_i$, then ${\rm OPT}\geq nr\cdot W_{\mathcal{R}}(S_i)/(1+\varepsilon')$ holds with at least $1-\delta_3$ probability.
\end{lem}
\begin{proof}
See Appendix D.
\end{proof}
\begin{thm}
	Given $\delta_3=1/(n^\ell\cdot\log_2(nr))$, the number of Multi-Samplings $|\mathcal{R}|$ returned by Algorithm \ref{a5} satisfies $|\mathcal{R}|\geq\theta^*$ with at least $1-1/n^\ell$ probability.
\end{thm}
\begin{proof}
	In \cite{2018arXiv180809363C}, Chen pointed out this theorem cannot be obtained directly by combining Lemma 7 and Lemma 8. The multi-Samplings generated in $i^{th}$ iteration are biased samples, because of the fact that the algorithm enters the $i^{th}$ iteration means that the size of collection $\mathcal{R}$ in $(i-1)^{th}$ iteration cannot satisfy the termination condition. The complete proof is in the appendix of \cite{2018arXiv180809363C}. Based on that, Theorem 4 is established.
\end{proof}

\begin{algorithm}[!t]
	\caption{\textbf{Revised-IMM $(G,k,r,\varepsilon,\ell)$}}\label{a6}
	\begin{algorithmic}[1]
		\renewcommand{\algorithmicrequire}{\textbf{Input:}}
		\renewcommand{\algorithmicensure}{\textbf{Output:}}
		\REQUIRE $G=(V,E)$, parameters $k$, $r$, $\varepsilon$ and $\ell$
		\ENSURE $\{S_p^*$, $W_{\mathcal{R}}(S_p^*)\}$
		\STATE $\mathcal{R}\leftarrow$ Sampling($G$, $k$, $r$, $\varepsilon$, $\ell$)
		\STATE $\{S_p^*, W_{\mathcal{R}}(S_p^*)\}\leftarrow$ NodeSelection($\mathcal{R}, k$)
		\STATE Return $\{S_p^*$, $W_{\mathcal{R}}(S_p^*)\}$
	\end{algorithmic}
\end{algorithm}

\subsection{Time Complexity}
In the rest of this section, we discuss the time complexity of Algorithm \ref{a5}. We can observe that the computational cost of Algorithm \ref{a5} mainly concentrates on the generation of Multi-Sampling. First, we need to analyze the time of generating a Multi-Sampling. At the high level, we use breath-first search from a feature node to visit each of its incoming neighbors until reaching a rumor node. Thus, the expected time needed to generate a Multi-Sampling is $\mathbb{E}[w(R)]$, where $w(R)$ denotes the number of edges in $G$ that are incoming edges to the nodes in $R$.

\begin{lem}
	Considering the objective function $f^i(\cdot)$ defined as Equation (7), we have
	\begin{equation}
	\mathbb{E}[w(R)]=\frac{m\cdot\sum_{i=1}^{r}{\rm OPT}^i}{nr}
	\end{equation}
	where ${\rm OPT}^i$ is the optimal value of function $f^i(\cdot)$ and $r$ is the number of features.
\end{lem}
\begin{proof}
	We denote by $\mathcal{H}(v^i)$ the collection of all possible Multi-Samplings for a feature node $v^i$. For any Multi-Sampling $R\in \mathcal{H}(v^i)$, we have 
	\begin{flalign}
	&\mathbb{E}[w(R)]=\frac{\sum_{i=1}^{r}\sum_{v^i\in V^i}\sum_{R\in\mathcal{H}(v^i)}\Pr[R]\cdot w(R)}{nr}\nonumber\\
	&=\frac{\sum_{i=1}^{r}\sum_{v^i\in V^i}\sum_{R\in\mathcal{H}(v^i)}\Pr[R]\cdot\sum_{(y^i,z^i)\in E^i}x(\{z^i\},R)}{nr}\nonumber\\
	&=\frac{\sum_{(y^i,z^i)\in E^i}\sum_{i=1}^{r}\sum_{v^i\in V^i}\sum_{R\in\mathcal{H}(v^i)}\Pr[R]\cdot x(\{z^i\},R)}{nr}\nonumber\\
	&=\frac{\sum_{(y^i,z^i)\in E^i}\sum_{i=1}^{r}f^i(\{z^i\})}{nr}\nonumber\\
	&\leq\frac{m\cdot\sum_{i=1}^{r}{\rm OPT}^i}{nr}\nonumber
	\end{flalign}
	The Lemma is proved.
\end{proof}
\begin{lem}
	Algorithm \ref{a4} runs in $O(r\cdot\sum_{R\in\mathcal{R}}|R|)$ time.
\end{lem}
\begin{proof}
The running time of Algorithm \ref{a4} can be derived directly from Eqaution (10).
\end{proof}

Shown as above, the total number of Multi-Samplings generated in Algorithm \ref{a5} is $(|\mathcal{R}|+|\mathcal{R'}|)$. We denote by $i'$ the ending iteration of the for-loop, we have $|\mathcal{R}|=\lambda'/x_{i'}$ and $|\mathcal{R'}|=\lambda^*/{\rm LB}$ where $x_{i'}\leq {\rm LB}\leq {\rm OPT}$. The expected number of Multi-Samplings generated in Algorithm \ref{a5} can be expressed as $\mathbb{E}[|\mathcal{R}|]=O((\lambda'+\lambda^*)/{\rm OPT})$, thus
\begin{equation}
\mathbb{E}[|\mathcal{R}|]=O\left(\frac{(nr)(k+\ell)\log n}{{\rm OPT}\cdot\varepsilon^2}\right)
\end{equation}
From above, we can know that the expected time of generating all Multi-Samplings in Algorithm \ref{a5} is $\mathbb{E}[\sum_{R\in\mathcal{R}}w(R)]$. Based on Theorem 3 in \cite{tang2015influence}, another property of martingale \cite{williams1991probability}, we have $\mathbb{E}[\sum_{R\in\mathcal{R}}w(R)]=\mathbb{E}[|\mathcal{R}|]\cdot\mathbb{E}[w(R)]$.
Thus, 
\begin{equation}
\mathbb{E}[\sum_{R\in\mathcal{R}}w(R)]=O((k+\ell)m\log n/\varepsilon^2)
\end{equation}
due to the fact that $\sum_{i=1}^{r}{\rm OPT}^i=O({\rm OPT})$. Besides, $\mathbb{E}[\sum_{R\in\mathcal{R}}|R|]\leq \mathbb{E}[\sum_{R\in\mathcal{R}}w(R)]$ because $|R|\leq w(R)$ for any $R\in\mathcal{R}$. Thus, the total running time is $O((k+\ell)mr\log n/\varepsilon^2)$. Then, we have the following theorem:
\begin{thm}
	Algorithm \ref{a6} can be ganranteed to return a $(1-1/e-\varepsilon)$-approximate solution of MFRB problem with at least $1-1/n^\ell$ probability, and runs in $O((k+\ell)mr\log n/\varepsilon^2)$ expected time.
\end{thm}
\begin{proof}
	In \cite{2018arXiv180809363C}, Chen pointed out a direct combination of Theorem 3 and Theorem 4 is problematic. For Theorem 3, it is correct given a fixed value of $\theta$, which means that these $\theta$ Multi-Samplings are sampled from the same sample space. Theorem 4 is based on the satisfaction of Theorem 3, and it uses the same base sample from the probability space. In section 2.4 of \cite{2018arXiv180809363C}, Chen proved its correctness of that and provided us with two solutions in section 2.5. Here, we choose the first solution for our MFRB problem: regenerating a new collection of Multi-Samplings. In line 18 of Algorithm \ref{a5}, after determining the size of $\theta$, we regenerate a new collection of Multi-Samplings with the length of $\theta$, from line 19 to line 22 of Algorithm \ref{a5}, and feed it into Algorithm \ref{a4} to get the final result. In section 3.1 of \cite{2018arXiv180809363C}, Chen proved that it is bounded with at least $1-1/n^\ell$ probability, which answered the question mentioned above why we need to generate a new collection of Multi-Samplings.
\end{proof}

\section{Experiment}
In this section, we will show the effectiveness and efficiency of our proposed algorithms on three real social networks. Our goal is to evaluate Algorithm 5 and Algorithm 6 with some common used baseline algorithms.

\subsection{Dataset description and Statistics}
Our experiments are relied on the datasets from networkrepository.com \cite{nr}, an online network repository.
There are three datasets used in our experiments: (1) Dataset-1: A co-authorship network, where each edge is a co-authorship among scientists to publish papers in the area of network theory. (2) A Wiki network, which is a who-voteson-whom network collected from Wikipedia. (3) Dataset-3: an Advogato online social network, which is a social community platform. Users can express weighted trust relationships among themselves explicitly. These datasets contain a list of all of the user-to-user links. Basic statistics of these datasets are summarized in Table 1. However, according to the multi-layer structure of MF-model, the number of feature nodes is dfferent from these basic information. Thus, the actual number of nodes and edges is determined by the number of features. we will describe in detail later.

\begin{figure}[!t]
	\centering
	\includegraphics[width=3.5in]{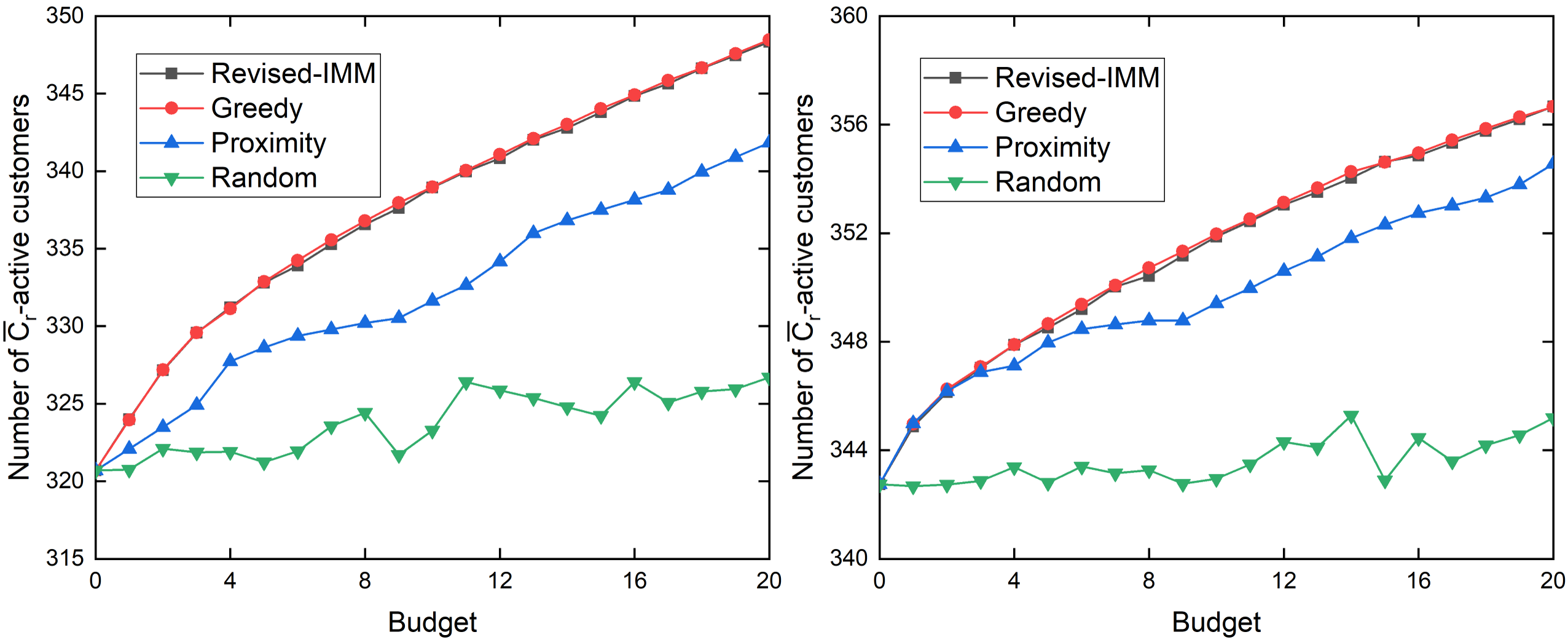}
	\\(a) sub-case: two features
	\includegraphics[width=3.5in]{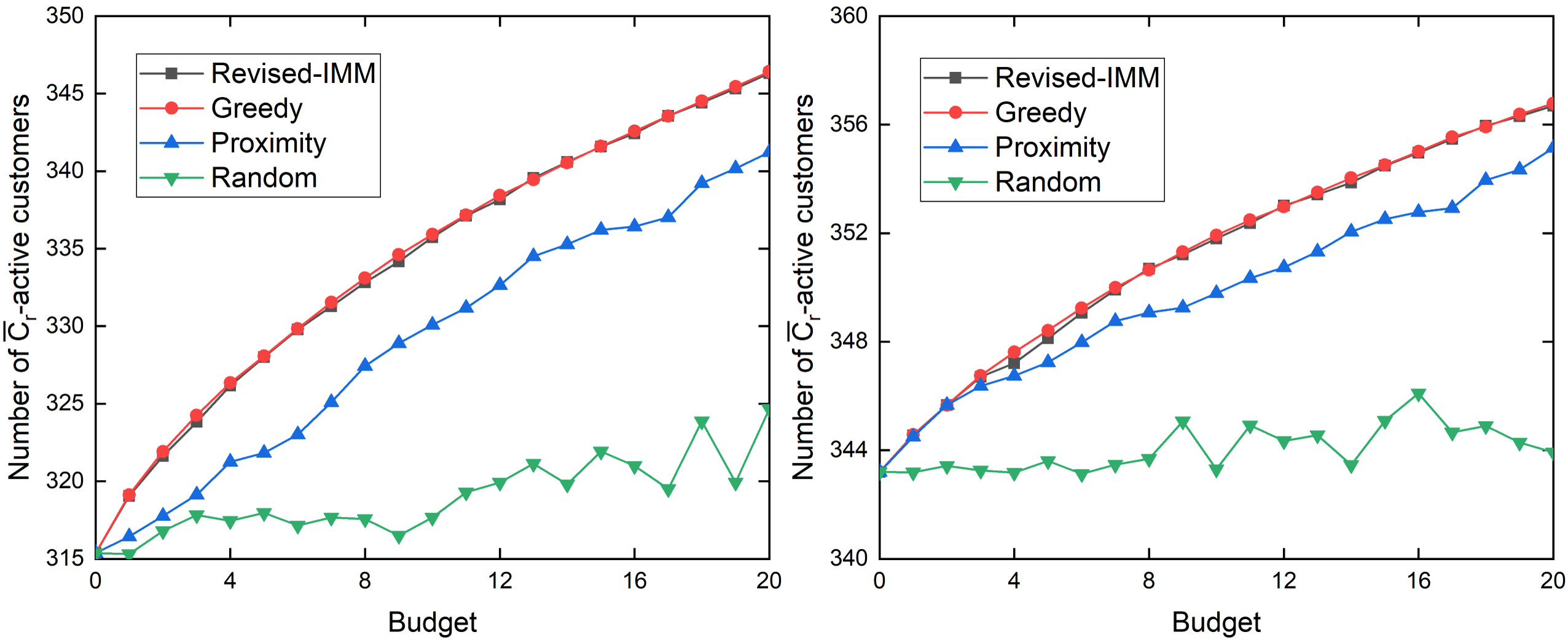}
	\\(b) sub-case: three features
	\includegraphics[width=3.5in]{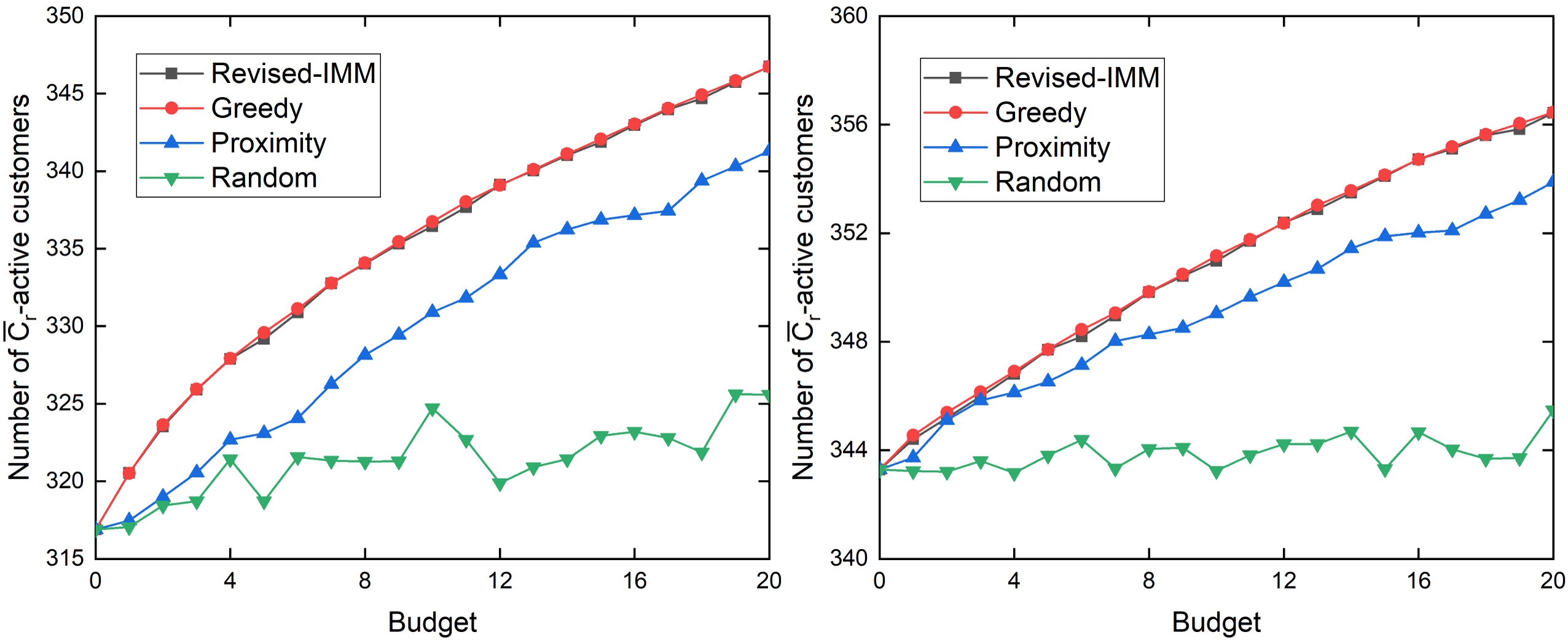}
	\\(c) sub-case: four features
	\caption{The performance comparison achieved by different algorithms with budget 20 in dataset-1. The left column is under the CP-model, and right column is under WC-model.}
	\label{fig2}
\end{figure}

\begin{figure}[!t]
	\centering
	\includegraphics[width=3.5in]{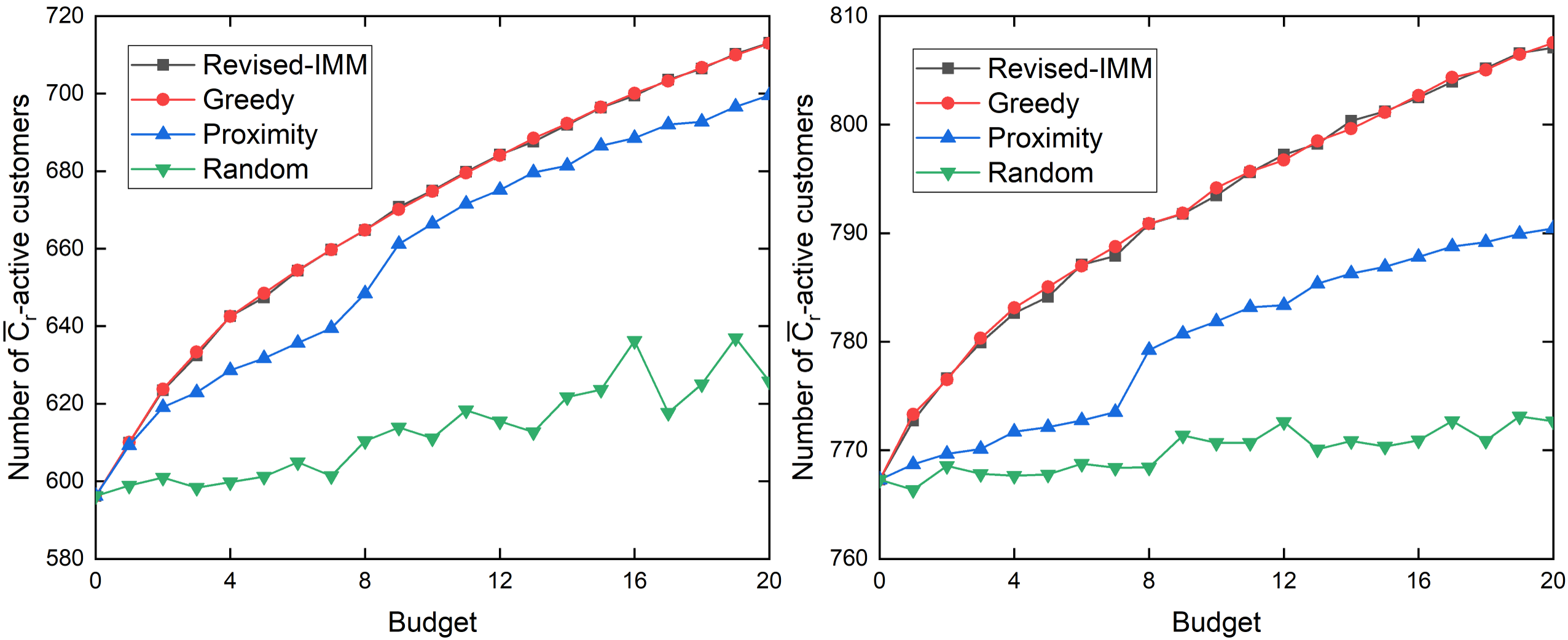}
	\\(a) sub-case: two features
	\includegraphics[width=3.5in]{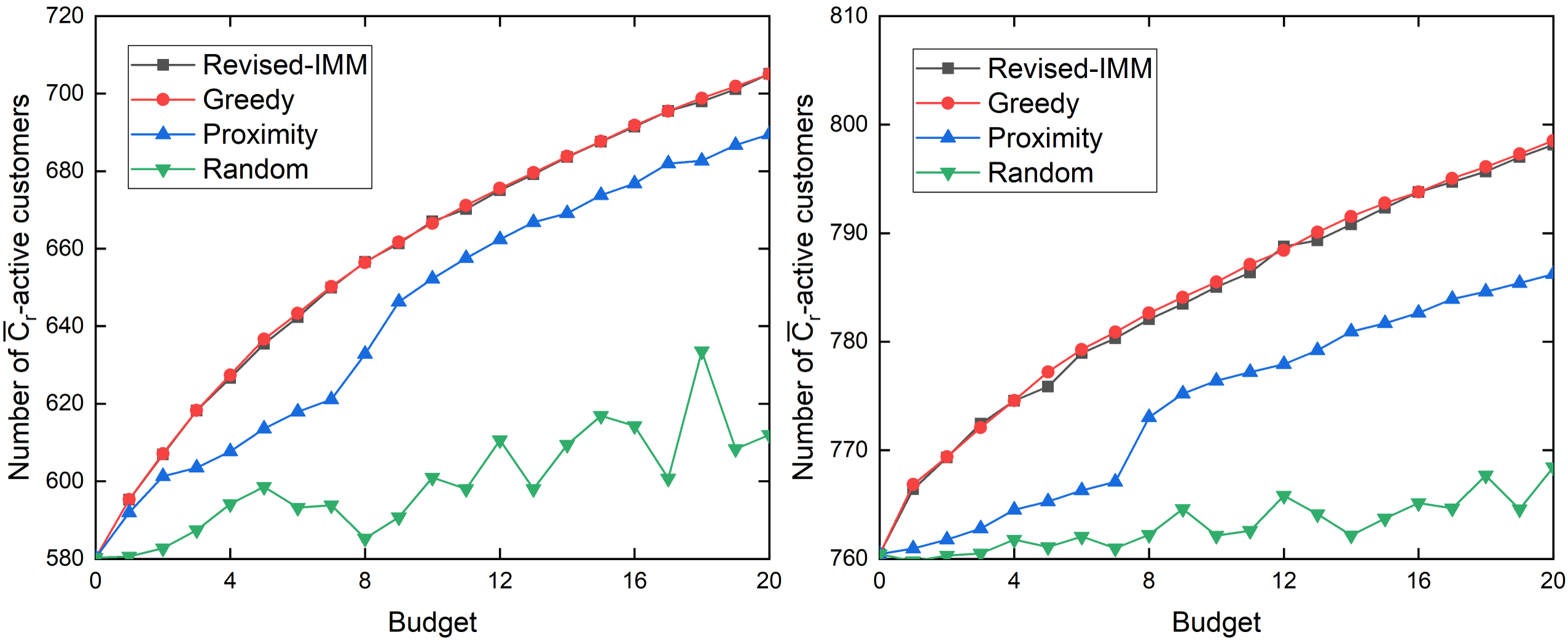}
	\\(b) sub-case: three features
	\includegraphics[width=3.5in]{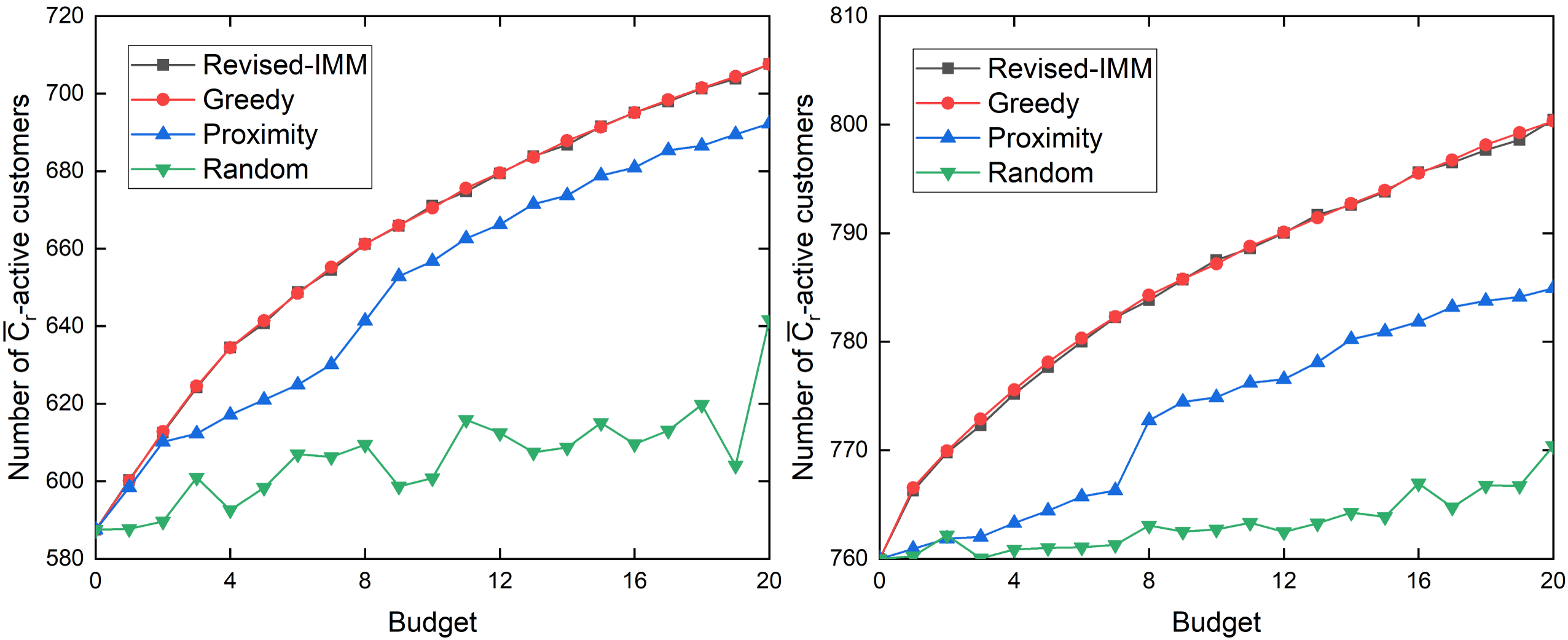}
	\\(c) sub-case: four features
	\caption{The performance comparison achieved by different algorithms with budget 20 in dataset-2. The left column is under the CP-model, and right column is under WC-model.}
	\label{fig3}
\end{figure}

\begin{table}[h]
	\renewcommand{\arraystretch}{1.3}
	\caption{The statistics of three datasets}
	\label{table_example}
	\centering
	\begin{tabular}{|c|c|c|c|c|}
		\hline
		\bfseries Dataset & \bfseries n & \bfseries m & \bfseries Type & \bfseries Average degree\\
		\hline
		dataset-1 & 0.4K & 1.01K & directed & 4\\
		\hline
		dataset-2 & 1.0K & 3.15K & directed & 6\\
		\hline
		dataset-3 & 6.5K & 51.3K & directed & 18\\
		\hline
	\end{tabular}
\end{table}

\subsection{Experimental Setup}
The experiment is based on MF-model, thus, the probability on the edges is either set as a constant or for each edge $e=(u,v)$, we set $p_e=1/|N^-(v)|$. This setting  is widely used in prior works \cite{tang2014influence} \cite{goyal2011celf++} \cite{jung2011irie}. We call these two setting as constant probability model (CP-model) and weighted cascade model (WC-model). Then, we compare our proposed algorithms with some common baseline algorithm, which is shown as follows:
\begin{itemize}
	\item Revised-IMM: This is the algorithm proposed in this paper, unless otherwise specified, we set $\varepsilon=0.1$ and $\ell=1$ by default.
	\item Greedy: At each step, it selects a node such that adding this node to current seed set can obtain the maximum marginal gain. It is implemented by Monte-Carlo simulation. We set the number of simulations to $num=2000$. It is only tested on small networks, dataset-1 and dataset-2, because of its low efficiency.
	\item Proximity: It selects the outgoing neighbors of the nodes in rumor set according to the out-degree of these outgoing neighbors. We select these neighbors with high out-degree in priority.
	\item Random: This is a classical baseline algorithm, where the nodes in positive set are selected randomly.
\end{itemize}

In our experiment, the users in rumor seed set $S_r$ are the nodes with the highest outgoing degree in original graph $G$ and the size $|S_r|=20$. Because the $S_r$ is partially $C_r$-active, only part of features of those users in $S_r$ are $C_r$-accepted, thus, we set the probability that the corresponding feature nodes of $S_r$ accept rumor cascade is $80\%$. The number of users in positive set $S_p$ is from $1$ to $20$, and $S_p$ is fully $C_p$-active, so the corresponding feature nodes of $S_p$ are all $C_p$-accepted. 

Next, we evaluate the performance of Revised-IMM algorithm. It can be divided into three sub-cases: (a) Assuming for each product, there are two features 1 and 2, the corresponding graph $G'$ has two layers, one is feature 1 and the other is feature 2. For CP-model, the diffusion probability for feature 1 is $p^1=0.4$ and feature 2 is $p^2=0.5$. The weight for feature 1 is $w^1=0.3$ and feature 2 is $w^2=0.7$. The actual number of nodes and edges will be doubled. (b) Assuming for each product, there are three features 1, 2 and 3, thus,  $G'$ has three layers for each feature. For CP-model, the diffusion probability for feature 1 is $p^1=0.4$, feature 2 is $p^2=0.5$ and feature 3 is $p^3=0.6$. The weight for feature 1 is $w^1=0.3$, feature 2 is $w^2=0.3$ and feature 3 is $w^3=0.4$. The actual number of nodes and edges will be tripled. (c) Assuming for each product, there are four features 1, 2, 3 and 4, thus, $G'$ has four layers for each feature. For CP-model, the diffusion probability for feature 1 is $p^1=0.4$, feature 2 is $p^2=0.5$, feature 3 is $p^3=0.5$ and feature 4 is $p^4=0.6$. The weight for feature 1 is $w^1=0.2$, feature 2 is $w^2=0.3$, feature 3 is $w^3=0.4$ and feature 4 is $w^4=0.1$. The actual number of nodes and edges will be quadrupled.

\begin{figure}[!t]
	\centering
	\includegraphics[width=3.5in]{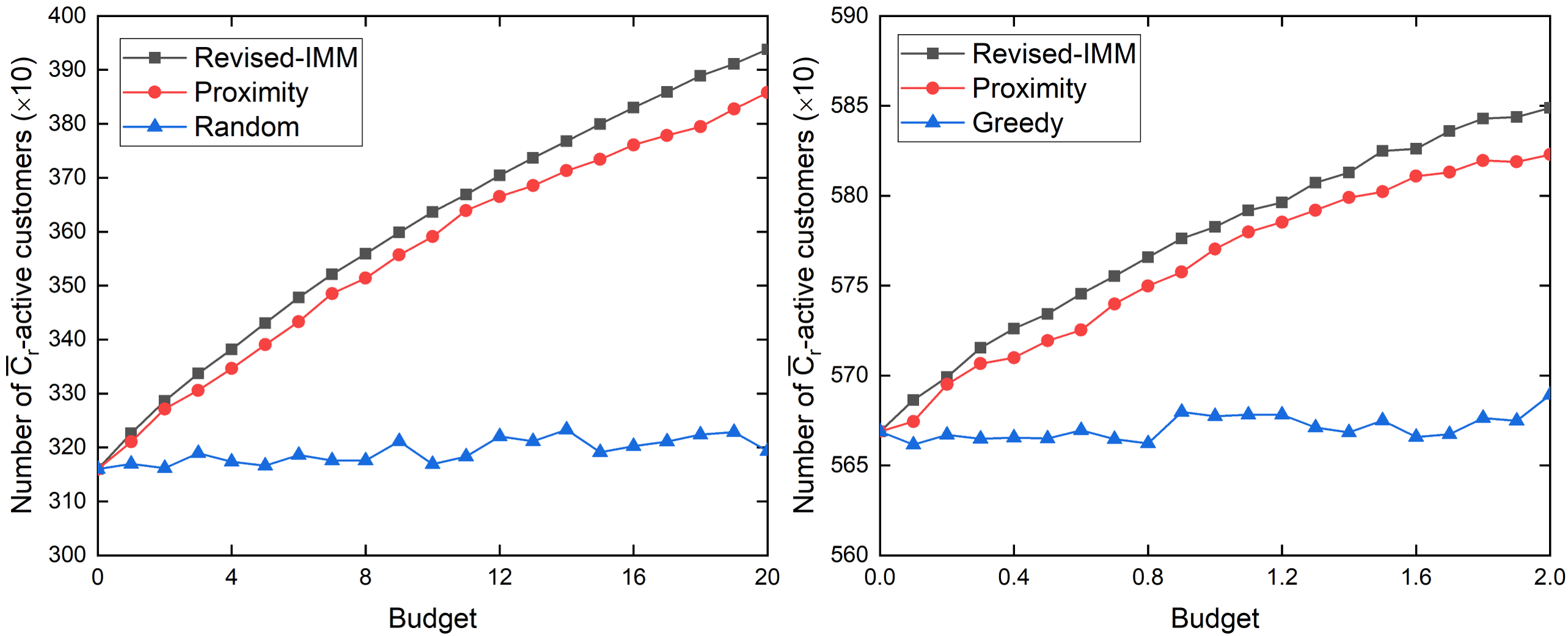}
	\\(a) sub-case: two features
	\includegraphics[width=3.5in]{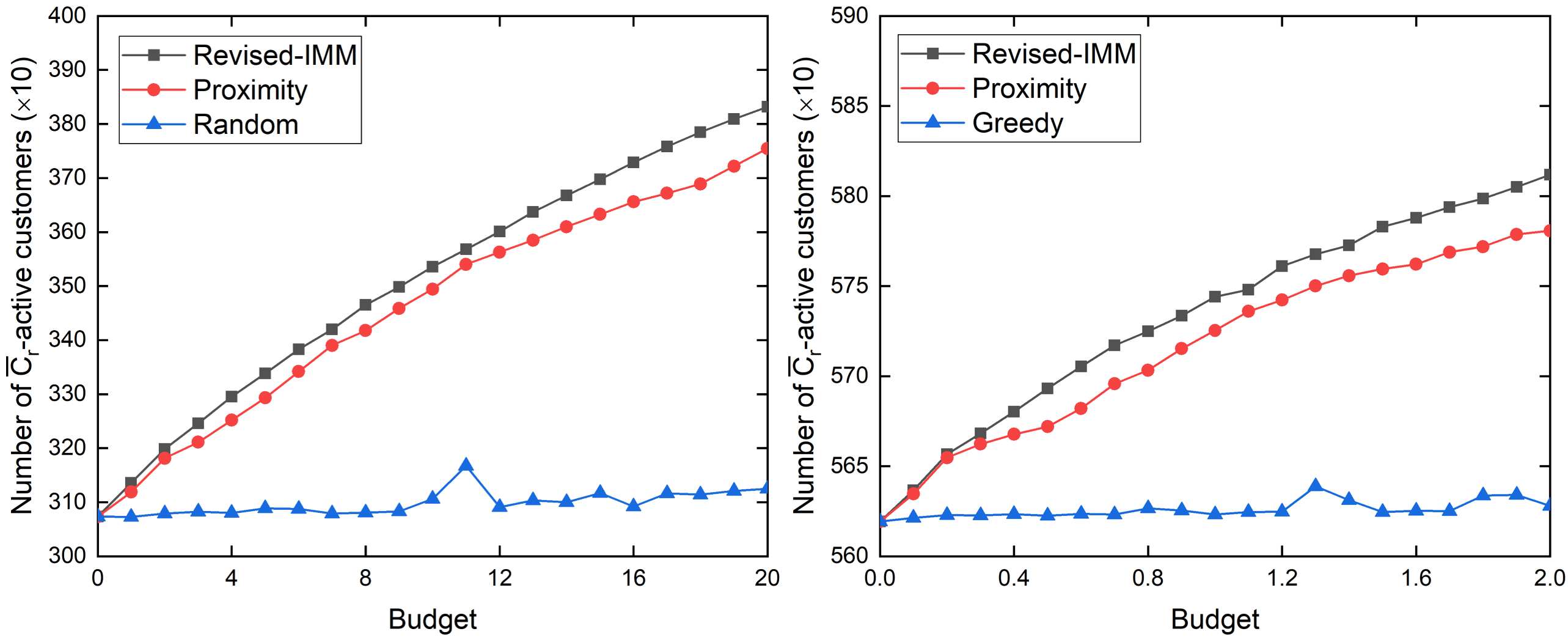}
	\\(b) sub-case: three features
	\caption{The performance comparison achieved by different algorithms with budget 20 in dataset-3. The left column is under the CP-model, and right column is under WC-model.}
	\label{fig4}
\end{figure}

\begin{table}[h]
	\renewcommand{\arraystretch}{1.3}
	\caption{The running time under the PC-model when $k=20$}
	\label{t2}
	\centering
	\begin{tabular}{|c|c|c|c|c|}
		\hline
		\multicolumn{5}{|c|}{Dataset-1}\\
		\hline
		\bfseries & \bfseries Revised-IMM & \bfseries Greedy & \bfseries Proximity & \bfseries Random\\
		\hline
		(a) & 16.67s & 1.67h & 0.74s & 0.89s\\
		\hline
		(b) & 14.02s & 2.78h & 1.80s & 1.91s\\
		\hline
		(c) & 30.42s & 3.77h & 2.06s & 2.33s\\
		\hline
		\multicolumn{5}{|c|}{Dataset-2}\\
		\hline
		(a) & 148.24s & 19.66h & 4.01s & 3.69s\\
		\hline
		(b) & 167.48s & 32.17h & 6.35s & 6.50s\\
		\hline
		(c) & 193.87s & 41.42h & 8.59s & 8.77s\\
		\hline
		\multicolumn{5}{|c|}{Dataset-3}\\
		\hline
		(a) & 14min & n/a & 2min & 1min\\
		\hline
		(b) & 20min & n/a & 2min & 1min\\
		\hline
	\end{tabular}
\end{table}

\subsection{Experimental results}
Figure \ref{fig2} and Figure \ref{fig3} draw the performance comparison achieved by different algorithms with budget 20 under the dataset-1 and dataset-2. Obviously, we can see that Revised-IMM algorithms and Greedy algorithm have the the same performance with respect to objective function $f(\cdot)$. However, to computational cost, Revised-IMM algorithm is much more efficient that Greedy algorithm. The running time in these experiment is shown as Table \ref{t2}. For example, under the CP-model with $k=20$, we consider dataset-2 with 4 features, Revised-IMM consumes $193.87$ seconds but Greedy takes about $41.42$ hours.

Figure \ref{fig4} draws the performance comparison achieved by different algorithms with budget 20 under the dataset-3. It verifies the scalability of Revised-IMM algorithm. Figure \ref{fig5} draws the number of Multi-Samplings generated by Algorithm \ref{a5} with different budgets. We can see that this is in line with our expectation, the number of Multi-Samplings increases as the budget increases. Figure \ref{fig6} draws the average relative error between the estimated value and objective value from budget 1 to 20. Here, given positive seed set $S_p$ and a collection of Multi-Samplings, the estimated value is $nr\cdot W_{\mathcal{R}}(S_p)$ and objective value is $f(S_p)$, which is implemented by Monte-Carlo simulation with $num=2000$. Thus, the relative error is $|f(S_p)-nr\cdot W_{\mathcal{R}}(S_p)|/f(S_p)$. For example, under the setting: dataset-1, $k=20$, CP-model and 2 features, estimated value is $348.38$ and objective value is $348.35$, we have relative error is $0.01\%$. Therefore, it satisfies what Theorem 2 said, $nr\cdot W_{\mathcal{R}}(S_p)$ is an unbiased estimator to $f(S_p)$.

\begin{figure}[!t]
	\centering
	\includegraphics[width=3.5in]{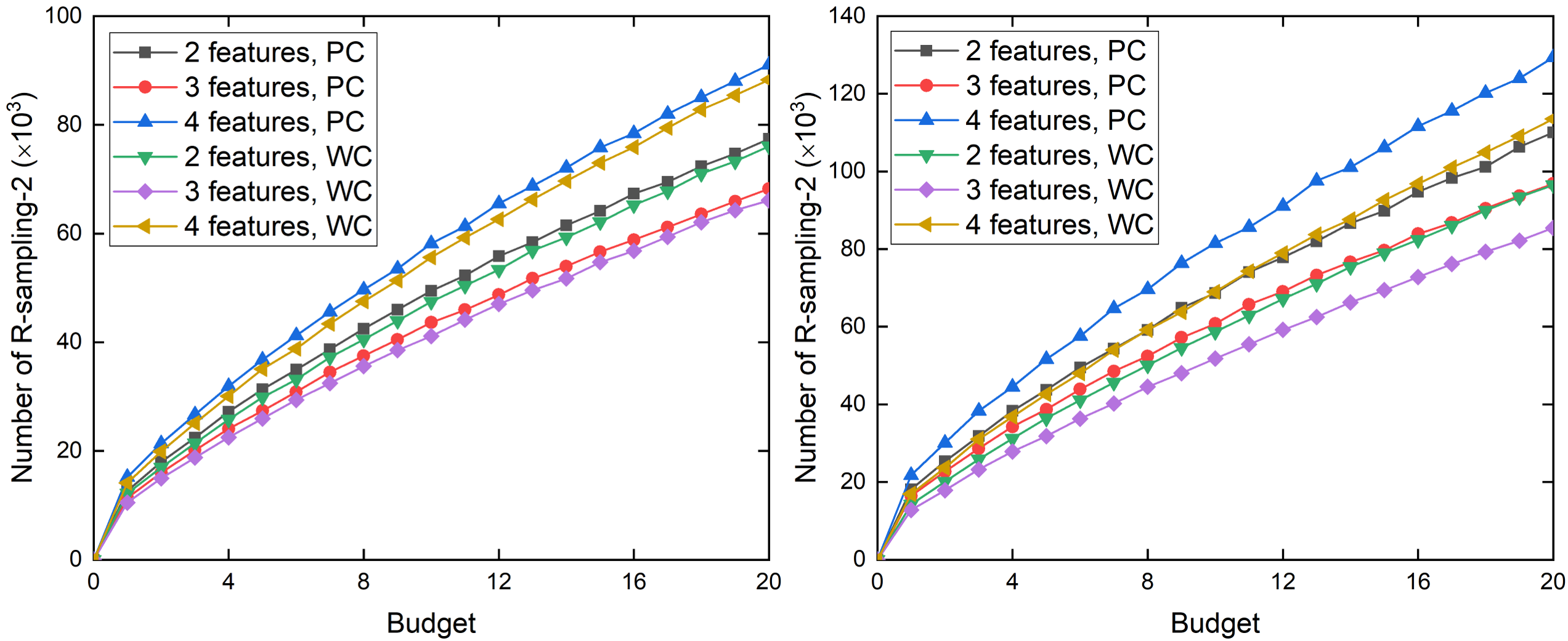}
	\caption{The number of Multi-Samplings generated by Algorithm \ref{a5} in Revised-IMM algorithm. The left column is under the dataset-1, and right column is under the dataset-2.}
	\label{fig5}
\end{figure}
\begin{figure}[!t]
	\centering
	\includegraphics[width=3.5in]{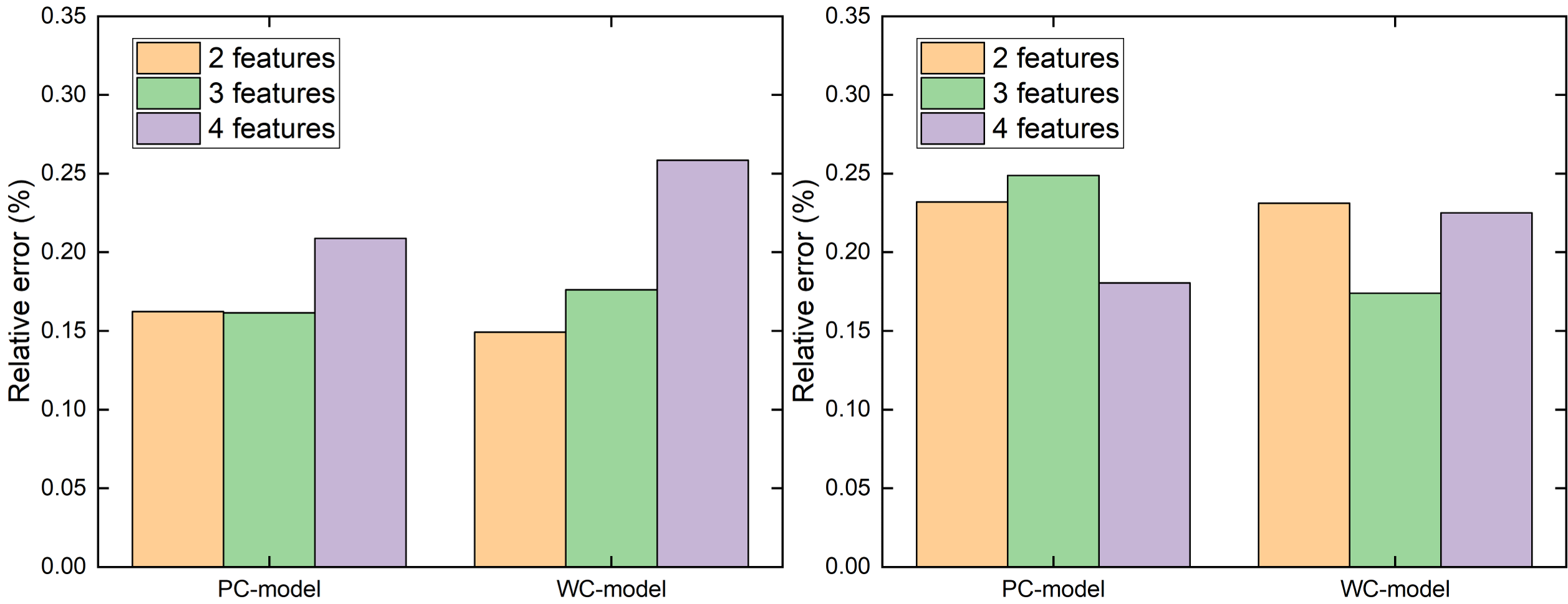}
	\caption{The avergae relative error between the estimated value and objective value from budget 1 to 20. Here, the left column is under the dataset-1, and right column is under the dataset-2.}
	\label{fig6}
\end{figure}

\section{Conclusion}
In this paper, we propose a novel multi-feature diffusion model, MF-model, to simulate real scenario in which multiple features can be propagated independently in social networks. Based on MF-model, MFRB problem is formulated as a monotone non-decreasing submodular maximization problem. Then, we design a novel sampling technique, Multi-Sampling, which is an unbiased estimator to objective function of MFRB. Inspired by martingale analysis, the Revised-IMM algorithm is proposed, which returns a $(1-1/e-\varepsilon)$-approximation solution and runs in $O((k+\ell)mr\log n/\varepsilon^2)$ expected time. The experimental result verified the effectiveness and correctness of Revised-IMM algorithm.

However, one of the shortcomings of this paper is that the weight for each feature is equal for different users, which is not entirely realistic. Because for different users, the importance of each feature to them is different. For example, some people care more about price, others value the appearance more. In future work, it is worth studying how to solve this more complicated and more realistic situation, which is not submodular, and even not monotone non-decreasing.


%

\appendix
\subsection{Proof of Lemma 4}
\begin{proof}
	For optimal solution $S_p^\circ$, we have defined $p=f(S_p^\circ)/nr$, thus, $p={\rm OPT}/nr=\mathbb{E}[W_{\mathcal{R}}(S_p^\circ)]$. Then, by Equation (16), we have
	\begin{flalign}
		&\Pr[nr\cdot W_{\mathcal{R}}(S_p^\circ)\leq (1-\varepsilon_1)\cdot {\rm OPT}]\nonumber\\
		&=\Pr[nr\cdot W_{\mathcal{R}}(S_p^\circ)\leq (1-\varepsilon_1)\cdot pnr]\nonumber\\
		&=\Pr[\theta\cdot W_{\mathcal{R}}(S_p^\circ)\leq (1-\varepsilon_1)\cdot p\theta]\nonumber\\
		&=\Pr\bigg[\sum_{j=1}^{\theta}\sum_{i=1}^{r}w^i\cdot x(S_p^i,R_j)-p\theta\leq -\varepsilon_1\cdot p\theta\bigg]\nonumber\\
		&\leq\exp\bigg(-\frac{\varepsilon_1^2}{2\bar{w}}\cdot p\theta\bigg)\nonumber\\
		&\leq\exp\bigg(-\frac{\varepsilon_1^2}{2\bar{w}}\cdot p\theta_1\bigg)\nonumber\\
		&=\delta_1 \nonumber
	\end{flalign}
	Thus, $nr\cdot W_{\mathcal{R}}(S_p^\circ)\geq (1-\varepsilon_1)\cdot {\rm OPT}$ holds with at least $1-\delta_1$ probability. By Lemma 3 and greedy properties, $nr\cdot W_{\mathcal{R}}(S_p^*)\geq (1-1/e)\cdot nr\cdot W_{\mathcal{R}}(S_p^\circ)\geq (1-1/e)(1-\varepsilon_1)\cdot {\rm OPT}$. The Lemma is proved.
\end{proof}
\subsection{Proof of Lemma 5}
\begin{proof}
	For any $k$-size seed set $S_p$, we have defined $p=f(S_p)/nr$, thus, $p=\mathbb{E}[W_{\mathcal{R}}(S_p)]$. Then, by Equation (17) and $\zeta=\varepsilon_2\cdot {\rm OPT}/pnr$, we have
	\begin{flalign}
		&\Pr[nr\cdot W_{\mathcal{R}}(S_p^*)-f(S_p)\geq \varepsilon_2\cdot {\rm OPT}]\nonumber\\
		&=\Pr[nr\cdot W_{\mathcal{R}}(S_p)-pnr\geq \varepsilon_2\cdot {\rm OPT}]\nonumber\\
		&=\Pr[\theta\cdot W_{\mathcal{R}}(S_p)-p\theta\geq\frac{\varepsilon_2\cdot {\rm OPT}}{pnr}\cdot p\theta]\nonumber\\
		&=\Pr\bigg[\sum_{j=1}^{\theta}\sum_{i=1}^{r}w^i\cdot x(S_p^i,R_j)-p\theta\geq \frac{\varepsilon_2\cdot {\rm {\rm OPT}}}{pnr}\cdot p\theta\bigg]\nonumber\\
		&\leq\exp\bigg(-\frac{\zeta^2}{2\bar{w}+\frac{2}{3}\zeta}\cdot p\theta\bigg)\nonumber\\
		&=\exp\bigg(-\frac{\varepsilon_2^2\cdot{{\rm OPT}}^2}{2\bar{w}pn^2r^2+\frac{2}{3}\varepsilon_2nr\cdot {\rm OPT}}\cdot\theta\bigg)\nonumber\\
		&\leq\exp\bigg(-\frac{\varepsilon_2^2\cdot{{\rm OPT}}^2}{2\bar{w}nr\cdot {\rm OPT}+\frac{2}{3}\varepsilon_2nr\cdot {\rm OPT}}\cdot\theta\bigg)\nonumber\\
		&\leq\exp\bigg(-\frac{\varepsilon_2^2\cdot {\rm OPT}}{(2\bar{w}+\frac{2}{3}\varepsilon_2)\cdot nr}\cdot\theta_2\bigg)\nonumber\\
		&=\delta_2/\tbinom{n-n_r}{k} \nonumber
	\end{flalign}
	Because there exists at most $\binom{n-n_r}{k}$ positive size-$k$ seed sets and by union bound, there is at least $1-\delta_2$ probability that no such $S_p^*$ that $nr\cdot W_{\mathcal{R}}(S_p^*)-f(S_p^*)\geq \varepsilon_2\cdot {\rm OPT}$. The Lemma is proved.
\end{proof}
\subsection{Proof of Lemma 6}
\begin{proof}
	For any $k$-size seed set $S_i$, we have defined $p=f(S_i)/nr$, thus, $p=\mathbb{E}[W_{\mathcal{R}}(S_i)]\leq{\rm OPT}/nr<x_i/nr$. Then, by Equation (17) and $\zeta=\frac{(1-\varepsilon')\cdot x_i}{pnr}-1$, we know that $\zeta>\varepsilon'\cdot x_i/(pnr)>\varepsilon'$, and we have
	\begin{flalign}
			&\Pr\left[nr\cdot W_{\mathcal{R}}(S_i)\geq(1+\varepsilon')\cdot x_i\right]\nonumber\\
			&=\Pr\left[\theta_i\cdot W_{\mathcal{R}}(S_i)-p\theta_i\geq\left(\frac{(1-\varepsilon')\cdot x_i}{pnr}-1\right)\cdot p\theta_i\right]\nonumber\\
			&\leq\exp\left(-\frac{\zeta^2}{2\bar{w}+\frac{2}{3}\zeta}\cdot p\theta_i\right)\nonumber\\
			&<\exp\left(-\frac{\varepsilon'^2\cdot x_i/(pnr)}{2\bar{w}+\frac{2}{3}\zeta}\cdot\frac{\left(2\bar{w}+\frac{2}{3}\varepsilon'\right)pnr\left(\log\binom{n-n_r}{k}/\delta_3\right)}{\varepsilon'^{2}\cdot x_i}\right)\nonumber\\
			&<\exp\left(-\log\tbinom{n-n_r}{k}/\delta_3\right)\nonumber\\
			&=\delta_3/\tbinom{n-n_r}{k}\nonumber
	\end{flalign}
	Because there is at least $1-\delta_3$ probability by union bound that no such $S_i$ that $nr\cdot W_{\mathcal{R}}(S_i)\geq(1+\varepsilon')\cdot x_i$. The Lemma is proved.
\end{proof}
\subsection{Proof of Lemma 7}
\begin{proof}
	For any $k$-size seed set $S_i$, we have defined $p=f(S_i)/nr$, thus, $p=\mathbb{E}[W_{\mathcal{R}}(S_i)]\leq{\rm OPT}/nr$. Then, by Equation (17) and $\zeta=\frac{\varepsilon'\cdot{\rm OPT}}{pnr}$, we have
	\begin{flalign}
		&\Pr[{\rm OPT}<nr\cdot W_{\mathcal{R}}(S_i)/(1+\varepsilon')]\nonumber\\
		&=\Pr[nr\cdot  W_{\mathcal{R}}(S_i)-{\rm OPT}\geq\varepsilon'\cdot {\rm OPT}]\nonumber\\
		&<\Pr\left[\theta_i\cdot W_{\mathcal{R}}(S_i)-p\theta_i\geq\frac{\varepsilon'\cdot {\rm OPT}}{pnr}\cdot p\theta_i\right]\nonumber\\
		&\leq\exp\left(-\frac{\zeta^2}{2\bar{w}+\frac{2}{3}\zeta}\cdot p\theta_i\right)\nonumber\\
		&=\exp\left(-\frac{\varepsilon'^2\cdot{\rm OPT}^2}{2\bar{w}pn^2r^2+\frac{2}{3}\varepsilon'nr\cdot {\rm OPT}}\cdot\theta_i\right)\nonumber\\
		&=\exp\left(-\frac{\varepsilon'^2\cdot{\rm OPT}^2}{2\bar{w}pnr\cdot {\rm OPT}+\frac{2}{3}\varepsilon'nr\cdot {\rm OPT}}\cdot\theta_i\right)\nonumber\\
		&\leq\exp\left(-\frac{\varepsilon'^2\cdot{\rm OPT}}{(2\bar{w}+\frac{2}{3}\varepsilon')\cdot nr}\cdot\theta_i\right)\nonumber\\
		&=\delta_3/\tbinom{n-n_r}{k}\nonumber
	\end{flalign}
	Because there is at least $1-\delta_3$ probability by union bound that no such $S_i$ that ${\rm OPT}<nr\cdot W_{\mathcal{R}}(S_i)/(1+\varepsilon')$. The Lemma is proved.
\end{proof}


\section*{Acknowledgment}
This work is partly supported by National Science Foundation under grant 1747818.

\ifCLASSOPTIONcaptionsoff
  \newpage
\fi



%
\bibliographystyle{IEEEtran}
\bibliography{references}

%




\end{document}